%% file: main.tex
\documentclass[letterpaper,twocolumn,10pt]{article}
\usepackage{usenix2019_v3}

\usepackage[first=0,last=9]{lcg}
\usepackage{colortbl}
\definecolor{Gray}{gray}{0.8}

\input{prestuff.tex}

\usepackage{amsthm,amssymb,diagbox}

\usepackage{titlesec}
\titlespacing\section{0pt}{4pt}{4pt}
\titlespacing\subsection{0pt}{3pt}{3pt}
\titlespacing\subsubsection{0pt}{2pt}{2pt}
\titleformat{\subsection}{\large\bfseries}{\thesubsection}{1em}{}

\usepackage{caption}
\newtheorem{prop}{Proposition}
\newcommand\matt{{\scaleobj{0.8}{\top}}}
\newcommand{\system}{\textsc{Adv$^\mathtt{2}$}\xspace}

\newcommand{\bx}{{x_\circ}}

\newcommand{\bbm}{{m_\circ}}

\newcommand{\ax}{{x_*}}
\newcommand{\ay}{c_t}
\newcommand{\tm}{m_t}
\newcommand{\am}{{m_*}}

\newcommand{\grd}{{\sc Grad}\xspace}
\newcommand{\cam}{{\sc Cam}\xspace}
\newcommand{\gcam}{{\sc GradCam}\xspace}
\newcommand{\mask}{{\sc Mask}\xspace}
\newcommand{\rts}{{\sc Rts}\xspace}
\newcommand{\rtsa}{{\sc Rts}$^\mathtt{A}$\xspace}

\newcommand{\ig}{{\small{\sc Ig}}\xspace}

\newcommand{\rnet}{ResNet\xspace}
\newcommand{\dnet}{DenseNet\xspace}

\newcommand{\imls}{IDLS\xspace}
\newcommand{\aid}{\textsc{Aid}\xspace}

\newcommand{\imlses}{IDLSes\xspace}

\newcommand{\pgd}{\textsc{Pgd}\xspace}
\newcommand{\adef}{\textsc{StAdv}\xspace}

\newcommand{\iou}{IoU\xspace}

\newcommand{\fs}{{\sc Fs}\xspace}

\begin{document}

\date{}

\title{Interpretable Deep Learning under Fire}

\author{}

\author{
{\rm Xinyang Zhang}$^\ast$  \quad {\rm Ningfei Wang}$^\star$  \quad {\rm Hua Shen}$^\ast$ \quad {\rm Shouling Ji}$^\dagger$ \quad {\rm Xiapu Luo}$^\ddagger$  \quad {\rm Ting Wang}$^\ast$\\
$^\ast$Pennsylvania State University \quad $^\star$University of California Irvine \\ $^\dagger$Zhejiang University and Alibaba-ZJU Joint Institute of Frontier Technologies\\
$^\ddagger$Hong Kong Polytechnic University
} 

\maketitle


\input{abstract.tex}

\input{introduction.tex}

\input{background.tex}
\input{attack.tex}
\input{evaluation.tex}

\input{discussion.tex}

\input{countermeasure.tex}

\input{literature.tex}

\input{conclusion.tex}

\section*{Acknowledgments}

This material is based upon work supported by the National Science Foundation under Grant No. 1846151 and 1910546. Any opinions, findings, and conclusions or recommendations expressed in this material are those of the author(s) and do not necessarily reflect the views of the National Science Foundation. Shouling Ji is partially supported by NSFC under No. 61772466 and U1836202, the Zhejiang Provincial Natural Science Foundation for Distinguished Young Scholars under No. LR19F020003, and the Provincial Key Research and Development Program of Zhejiang, China under No. 2017C01055. Xiapu Luo is partially supported by Hong Kong RGC Project (No. PolyU 152279/16E, CityU C1008-16G).

\bibliographystyle{plain}
\bibliography{main}

\input{appendix.tex}

\end{document}

%% file: prestuff.tex
\usepackage{epsfig,amsmath,amsfonts,epsfig,multirow,makecell,caption,soul,csquotes,color,wrapfig,subcaption,mathtools,bm,spverbatim,booktabs,tcolorbox}
\usepackage[e]{esvect}

\makeatletter
\newif\if@restonecol
\makeatother

\usepackage[boxed, ruled, vlined, linesnumbered]{algorithm2e}
\SetKwRepeat{Do}{do}{while}

\newenvironment{changemargin}[2]{\begin{list}{}{
	\setlength{\topsep}{0pt}\setlength{\leftmargin}{0pt}
	\setlength{\rightmargin}{0pt}
	\setlength{\listparindent}{\parindent}
	\setlength{\itemindent}{\parindent}
	\setlength{\parsep}{0pt plus 1pt}
	\addtolength{\leftmargin}{#1}\addtolength{\rightmargin}{#2}
	}\item}
	{\end{list}}

\newenvironment{mitemize}{
	\begin{changemargin}{-3pt}{-0cm}
	\vspace{-10pt}
	\hspace{-5pt}
	\begin{itemize}
	\setlength{\itemsep}{3pt}}
	{\end{itemize}
	\vspace{2pt}
	\end{changemargin}}

\newcommand{\ssup}[2]{{#1}^{\scaleobj{0.8}{#2}}}
\newcommand{\ssub}[2]{{#1}_{\scaleobj{0.8}{#2}}}
\newcommand{\sboth}[3]{{#1}_{\scaleobj{0.8}{#2}}^{\scaleobj{0.8}{#3}}}

\usepackage{hyperref}

\newcommand{\msec}[1]{\S\,\ref{#1}}
\newcommand{\mref}[1]{\,\ref{#1}}
\newcommand{\meq}[1]{Eqn\,(\ref{#1})}
\newcommand{\mcite}[1]{\cite{#1}}

\usepackage{scalerel}[2016/12/29]

\makeatletter
\providecommand{\leadsfrom}{%
  \mathrel{\mathpalette\reflect@squig\relax}%
}
\newcommand{\reflect@squig}[2]{%
  \reflectbox{$\m@th#1\leadsto$}%
}
\makeatother

\newcommand{\dnn}{DNN\xspace}
\newcommand{\dnns}{DNNs\xspace}

\def\eqref#1{equation~\ref{#1}}

\def\1{\bm{1}}

\DeclareMathAlphabet{\mathsfit}{\encodingdefault}{\sfdefault}{m}{sl}
\SetMathAlphabet{\mathsfit}{bold}{\encodingdefault}{\sfdefault}{bx}{n}

\DeclareMathOperator{\sign}{sgn}

%% file: abstract.tex
\begin{abstract}

Providing explanations for deep neural network (\dnn) models is crucial for their use in security-sensitive domains. A plethora of interpretation models have been proposed to help users understand the inner workings of \dnns: how does a \dnn arrive at a specific decision for a given input? The improved interpretability is believed to offer a sense of security by involving human in the decision-making process. Yet, due to its data-driven nature, the interpretability itself is potentially susceptible to malicious manipulations, about which little is known thus far.

Here we bridge this gap by conducting the first systematic study on the security of interpretable deep learning systems (\imlses). We show that existing \imlses are highly vulnerable to adversarial manipulations. Specifically, we present \system, a new class of attacks that generate adversarial inputs not only misleading target \dnns but also deceiving their coupled interpretation models. Through empirical evaluation against four major types of \imlses on benchmark datasets and in security-critical applications (e.g., skin cancer diagnosis), we demonstrate that with \system the adversary is able to arbitrarily designate an input's prediction and interpretation. Further, with both analytical and empirical evidence, we identify the prediction-interpretation gap as one root cause of this vulnerability -- a \dnn and its interpretation model are often misaligned, resulting in the possibility of exploiting both models simultaneously. Finally, we explore potential countermeasures against \system, including leveraging its low transferability and incorporating it in an adversarial training framework. Our findings shed light on designing and operating \imlses in a more secure and informative fashion, leading to several promising research directions.

\end{abstract}


%% file: introduction.tex
\section{Introduction}
\label{sec:introduction}

The recent advances in deep learning have led to breakthroughs in many long-standing machine learning tasks (e.g., image classification\mcite{he:resnet}, natural language processing\mcite{sutskever:nips:2014}, and even playing Go\mcite{silver:nature:alphago}), enabling use cases previously considered strictly experimental.

However, the state-of-the-art performance of deep neural network (\dnn) models is often achieved at the cost of interpretability. It is challenging to intuitively understand the inference of complicated {\dnns} -- how does a \dnn arrive at a specific decision for a given input -- due to their high non-linearity and nested architectures. This is a major drawback for applications in which the interpretability of decisions is a critical prerequisite, while simple black-box predictions cannot be trusted by default. Another drawback of \dnns is their inherent vulnerability to adversarial inputs -- maliciously crafted samples to trigger target \dnns to malfunction\mcite{szegedy:iclr:2014,carlini:cw,kurakin:advbim} -- which  leads to unpredictable model behaviors and hinders their use in security-sensitive domains.

\begin{figure}
\centering
\epsfig{file = 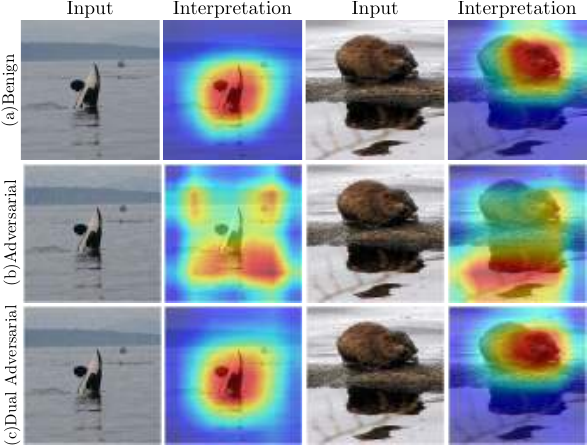, width=70mm}
\caption{\small Sample (a) benign, (b) regular adversarial, and (c) dual adversarial inputs and interpretations on {\rnet}\mcite{he:resnet} (classifier) and {\sc Cam}\mcite{zhou:cam} (interpreter). \label{fig:sample}}
\end{figure}

The drawbacks have spurred intensive research on improving the \dnn interpretability via providing explanations at either model-level\mcite{Karpathy:2016:iclr,Sabour:2017:nips,zhang:2018:interpretable} or instance-level\mcite{Simonyan:gradsaliency,fong:mask,tulio:sigkdd:2016,Dabkowski:nips:2017}. For example, in Figure\mref{fig:sample}\,(a), an attribution map highlights an input's most informative part with respect to its classification, revealing their causal relationship. Such interpretability helps users understand the inner workings of \dnns, enabling use cases including model debugging\mcite{Nguyen:2014:cvpr}, digesting security analysis results\mcite{Guo:2018:ccs}, and detecting adversarial inputs\mcite{Du:2018:kdd}. For instance, in Figure\mref{fig:sample}\,(b), an adversarial input, which causes the target \dnn to deviate from its normal behavior, generates an attribution map highly distinguishable from its benign counterpart, and is thus easily detectable.

As illustrated in Figure\mref{fig:imls}, a \dnn model (classifier), coupled with an interpretation model (interpreter), forms an interpretable deep learning system (\imls). The enhanced interpretability of \imlses is believed to offer a sense of security by involving human in the decision process\mcite{Tao:nips:2018}. However, given its data-driven nature, this interpretability itself is potentially susceptible to malicious manipulations. Unfortunately, thus far, little is known about the security vulnerability of \imlses, not to mention mitigating such threats.

\begin{figure}
\centering
\epsfig{file = 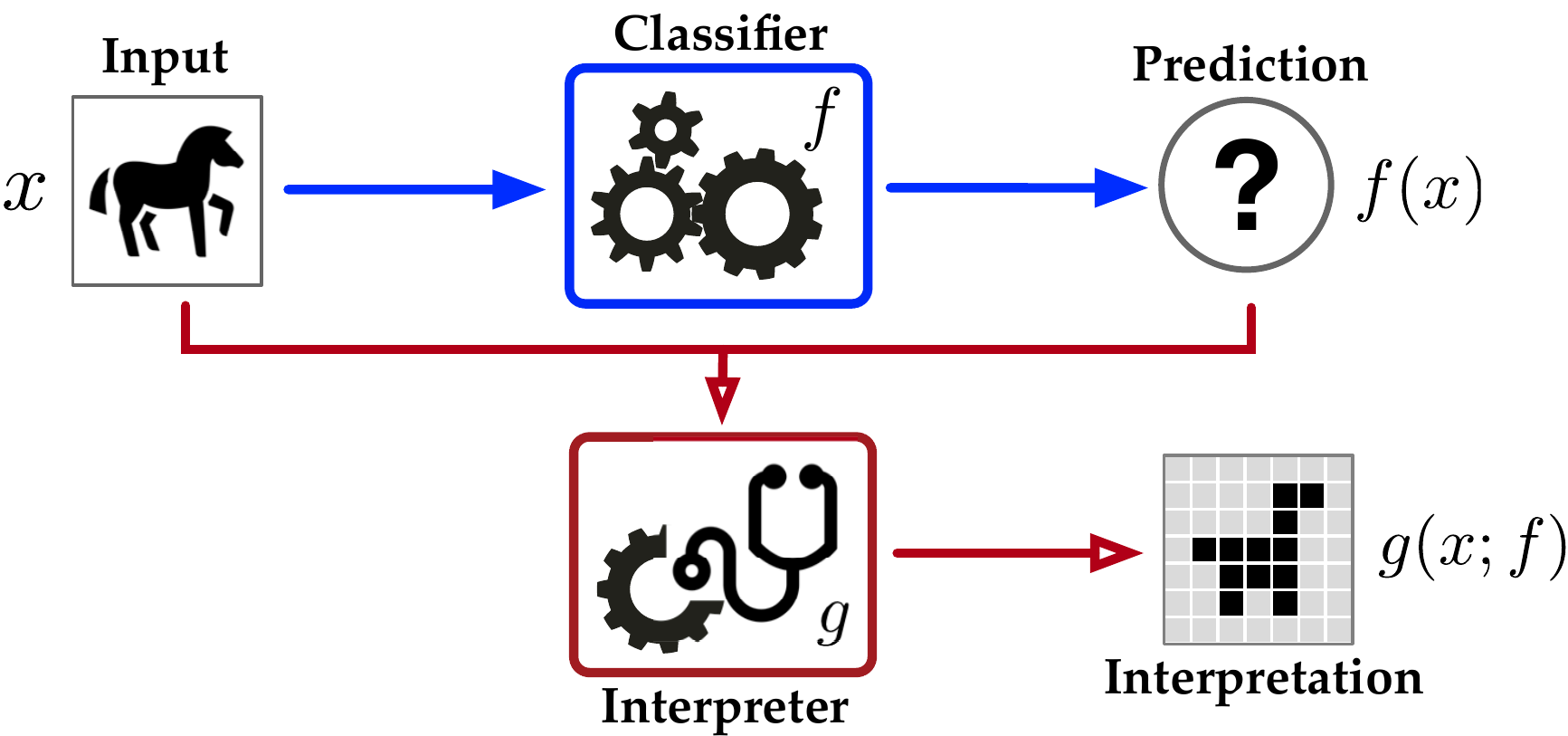, width=74mm}
\caption{\small Workflow of an interpretable deep learning system (\imls). \label{fig:imls}
}
\end{figure}


\vspace{2pt}
{\bf Our Work.} To bridge the gap, in this paper, we conduct a comprehensive study on the security vulnerability of \imlses, which leads to the following interesting findings.

First, we demonstrate that existing \imlses are highly vulnerable to adversarial manipulations. We present \system,
a new class of attacks that generate adversarial inputs not only misleading a target DNN but also deceiving its coupled interpreter. By empirically evaluating \system against four major types of \imlses on benchmark datasets and in security-critical applications (e.g., skin cancer diagnosis), we show that it is practical to generate adversarial inputs with predictions and interpretations arbitrarily chosen by the adversary. For example, Figure\mref{fig:sample}\,(c) shows adversarial inputs that are misclassified by target \dnns and also interpreted highly similarly to their benign counterparts. Thus the interpretability of \imlses merely provides limited security assurance.

Then, we show that one possible root cause of this attack vulnerability lies in the prediction-interpretation gap: the interpreter is often misaligned with the classifier, while the interpreter's interpretation only partially explains the classifier's behavior, allowing the adversary to exploit both models simultaneously. This finding entails several intriguing questions: (i) what, in turn, is the possible cause of this gap? (ii) how does this gap vary across different interpreters? (iii) what is its implication for designing more robust interpreters? We explore all these key questions in our study.

Further, we investigate the transferability of \system across different interpreters. We note that it is often difficult to find adversarial inputs transferable across distinct types of interpreters, as they generate interpretations from complementary perspectives (e.g., back-propagation, intermediate representations, input-prediction correspondence). This finding points to training an ensemble of interpreters as one potential countermeasure against \system.

Finally, we present adversarial interpretation distillation (\aid), an adversarial training framework which integrates \system in training interpreters. We show that \aid effectively reduces the prediction-interpretation gap and potentially helps improve the robustness of interpreters against \system.

To our best knowledge, this work represents the first systematic study on the security vulnerability of existing \imlses. We believe our findings shed light on designing and operating \imlses in a more secure and informative manner.

\vspace{2pt}
{\bf Roadmap.} The remainder of the paper proceeds as follows. \msec{sec:ground} introduces fundamental concepts; \msec{sec:attack} presents the \system attack and details its implementation against four major types of interpreters; \msec{sec:eval} empirically evaluates its effectiveness; \msec{sec:tradeoff} explores the fundamental causes of the attack vulnerability and discusses possible countermeasures; \msec{sec:literature} surveys relevant literature; the paper is concluded in \msec{sec:conclusion}.

%% file: background.tex
\begin{table}{\footnotesize
\centering
			\begin{tabular}{c | l }
				Notation & Definition\\
				\hline
				\hline
				$f$, $g$ & target classifier, interpreter \\
				$\bx$, $\ax$ & benign, adversarial input \\
				$\ay$, $\tm$ & adversary's target class, interpretation\\
        $x[i]$ & $i$-th dimension of $x$\\
				$\epsilon$ & perturbation magnitude bound\\
				$\| \cdot \|$ & vector norm\\
				$\ell_\mathrm{int}$, $\ell_\mathrm{prd}$, $\ell_\mathrm{adv}$   & interpretation, prediction, overall loss\\
				$\alpha$ & learning rate\\
				\hline
			\end{tabular}
    \caption{\small Symbols and notations. \label{tab:notations}}}
\end{table}

\section{Preliminaries}
\label{sec:ground}

We begin with introducing a set of fundamental concepts and assumptions. The symbols and notations used in this paper are summarized in Table\mref{tab:notations}.

\vspace{2pt}
{\bf Classifier} -- In this paper, we primarily focus on predictive tasks (e.g., image classification\mcite{Deng:2009:cvpr}), in which a \dnn $f$ (i.e., classifier) assigns a given input $x$ to one of a set of predefined classes $\mathcal{C}$, $f(x) = c\in \mathcal{C}$.

\vspace{2pt}
{\bf Interpreter} -- In general, the DNN interpretability can be obtained in two ways: designing interpretable {\dnns}\mcite{Zhang:2018:cvpr,Sabour:2017:nips} or extracting post-hoc interpretations. The latter case does not require modifying model architectures or parameters, thereby leading to higher prediction accuracy. We thus mainly consider post-hoc interpretations in this paper. More specifically, we focus on instance-level interpretability\mcite{Karpathy:2016:iclr,Sabour:2017:nips,Murdoch:2017:iclr,Dabkowski:nips:2017,Shrikumar:2017:icml,zhang:2018:interpretable,Simonyan:gradsaliency,fong:mask,Murdoch:2018:iclr}, which explains how a \dnn $f$ classifies a given input $x$ and uncovers the causal relationship between $x$ and $f(x)$. We assume such interpretations are given in the form of {\em attribution maps}. As shown in Figure\mref{fig:imls}, the interpreter $g$ generates an attribution map $m = g(x; f)$, with its $i$-th element $m[i]$ quantifying the importance of $x$'s $i$-th feature $x[i]$ with respect to $f(x)$.

\vspace{2pt}
{\bf Adversarial Attack} -- \dnns are inherently vulnerable to adversarial inputs, which are maliciously crafted samples to force \dnns to misbehave\mcite{szegedy:iclr:2014,moosavi:cvpr:2017}. Typically, an adversarial input $\ax$ is generated by modifying a benign input $\bx$ via pixel perturbation (e.g., {\pgd}\mcite{madry:towards}) or spatial transformation (e.g., {\adef}\mcite{Xiao:2018:iclr}), with the objective of forcing $f$ to misclassify $\ax$ to a target class  $\ay$, $f(\ax) = \ay \neq f(\bx)$.
To ensure the attack evasiveness, the modification is often constrained to an allowed set (e.g., a norm ball $\mathcal{B}_\epsilon(\bx) = \{x | \|x - \bx \|_\infty \leq \epsilon\}$).
Consider {\pgd}, a universal first-order adversarial attack, as a concrete case.
At a high level, \pgd implements a sequence of project gradient descent on the loss function:
\begin{align}
	\label{eq:pgd}
\ssup{x}{(i+1)} = \Pi_{\mathcal{B}_\epsilon(\bx)} \left(\ssup{x}{(i)} - \alpha \sign\left(\nabla_x \ell_\mathrm{prd}\left(f\left(\ssup{x}{(i)}\right), \ay\right)\right)\right)
\end{align}
where $\Pi$ is the projection operator, $\alpha$ represents the learning rate, the loss function $\ell_\mathrm{prd}$ measures the difference of the model prediction $f(x)$ and the class $\ay$ targeted by the adversary (e.g., cross entropy), and $\ssup{x}{(0)}$ is initialized as $\bx$.


\vspace{2pt}
{\bf Threat Model} -- Following the line of work on adversarial attacks\mcite{szegedy:iclr:2014,goodfellow:fsgm,carlini:cw,madry:towards}, we assume in this paper a white-box setting: the adversary has complete access to the classifier $f$ and the interpreter $g$, including their architectures and parameters. This is a conservative and realistic assumption. Prior work has shown that it is possible to train a surrogate model $f'$ given black-box access to a target \dnn $f$\mcite{papernot2017practical}; given that the interpreter is often derived directly from the classifier (details in \msec{sec:attack}), the adversary may then train a substitution interpreter $g'$ based on $f'$. We consider investigating such black-box attacks as our ongoing work.

%% file: attack.tex
\section{ADV$^2$ Attack}
\label{sec:attack}



The interpretability of \imlses is believed to offer a sense of security by involving human in the decision process\mcite{Tao:nips:2018,Guo:2018:ccs,Gehr:2018:sp,Du:2018:kdd}; this belief has yet to be rigorously tested. We bridge this gap by presenting \system, a new class of attacks that deceive target \dnns and their interpreters simultaneously. Below we first give an overview of \system and then detail its instantiations against four major types of interpreters.

\subsection{Attack Formulation}


The \system attack deceives both the \dnn $f$ and its coupled interpreter $g$. Specifically, \system generates an adversarial input $\ax$ by modifying a benign input $\bx$ such that
\begin{mitemize}
\item (i) $\ax$ is misclassified by $f$ to a target class $\ay$, $f(\ax) = \ay$;
\item (ii) $\ax$ triggers $g$ to generate a target attribution map $\tm$, $g(\ax; f) = \tm$;
\item (iii) The difference between $\ax$ and $\bx$, $\Delta(\ax, \bx)$, is imperceptible; 
\end{mitemize}
where the distance function $\Delta$ depends on the concrete modification: for pixel perturbation (e.g.,\mcite{madry:towards}), it is  instantiated as $\mathcal{L}_\mathrm{p}$ norm, while for spatial transformation (e.g.,\mcite{Xiao:2018:iclr}), it is defined as the overall spatial distortion.

In other words, the goal is to find sufficiently small perturbation to the benign input that leads to the prediction and interpretation desired by the adversary.

%


At a high level, we formulate \system using the following optimization framework:
\begin{align}
	\label{eq:optatt}
\min_{x}\,\, \Delta(x, \bx)
\quad
	\mathrm{s.t.}\,\, \left\{
	\begin{array}{l}
		f(x)  = \ay\\
		g(x;f) = \tm
	\end{array}
	\right.
\end{align}
where the constraints ensure that (i) the adversarial input is misclassified as $\ay$ and (ii) it triggers $g$ to generate the target attribution map $\tm$.

%
%

%

As the constraints of $f(x) = \ay$ and $g(x;f) = \tm$ are highly non-linear for practical \dnns, we reformulate \meq{eq:optatt} in a form more suited for optimization:
\begin{align}
	\label{eq:optatt2}
\nonumber \min_{x} &  \quad \ell_\mathrm{prd}(f(x), \ay) +  \lambda \ell_\mathrm{int} \left(g(x; f), \tm\right)\\
	\mathrm{s.t.} &  \quad 
			\Delta (x, \bx) \leq \epsilon
\end{align}
where the prediction loss $\ell_\mathrm{prd}$ is the same as in \meq{eq:pgd}, the interpretation loss $\ell_\mathrm{int}$ measures the difference of adversarial map $g(x;f)$ and target map $\tm$, and the hyper-parameter $\lambda$ balances the two factors. Below we use $\ell_\mathrm{adv}(x)$ to denote the overall loss function defined in \meq{eq:optatt2}.

We construct the solver of \meq{eq:optatt2} upon an adversarial attack framework. While it is flexible to choose the concrete framework, below we primarily use {\pgd}\mcite{madry:towards} as the reference and discuss the construction of \system upon alternative frameworks (e.g., spatial transformation\mcite{Xiao:2018:iclr}) in \msec{sec:eval}.

Under this setting, we define
$\ell_\mathrm{prd}(f(x), \ay) = -\log(f_{c_t}(x))$ (i.e., the negative log likelihood of $x$ with respect to the class $c_t$),
$\Delta(x, \bx) = \|x - \bx\|_\infty$, and $\ell_\mathrm{int}(g(x; f), \tm) = \|g(x; f) - \tm \|_2^2$.
In general, \system searches for $\ax$ using a sequence of gradient descent updates:
\begin{equation}
	\label{eq:update}
\ssup{x}{(i+1)}   =   \Pi_{\mathcal{B}_\epsilon(\bx)} \left(\ssup{x}{(i)} - \alpha \sign \left(\nabla_x  \ell_\mathrm{adv}\left(\ssup{x}{(i)}\right)\right)\right)
\end{equation}



However, directly applying \meq{eq:update} is often found ineffective, due to the unique characteristics of individual interpreters. In the following, we detail the instantiations of \system against the back-propagation-, representation-, model-, and perturbation-guided interpreters, respectively.


\input{method.tex}

\input{implementation.tex}

%% file: method.tex



\subsection{Back-Propagation-Guided Interpretation}

This class of interpreters compute the gradient (or its variants) of the model prediction with respect to a given input to derive the importance of each input feature. The hypothesis is that larger gradient magnitude indicates higher relevance of the feature to the prediction. We consider gradient saliency ({\grd})\mcite{Simonyan:gradsaliency} as a representative of this class.


Intuitively, \grd considers a linear approximation of the model prediction (probability) $f_c(x)$ for a given input $x$ and a given class $c$, and derives the attribution map $m$ as:
\begin{equation}
	\label{eq:gradient}
		m = \left| \frac{\partial f_c(x) }{\partial x} \right|
\end{equation}

To attack \grd-based \imlses, we may search for $\ax$ using a sequence of gradient descent updates as defined in \meq{eq:update}.
However, according to \meq{eq:gradient}, computing the gradient of the attribution map $g(x;f)$ amounts to computing the Hessian matrix of $f_c(x)$, which is all-zero for \dnns with ReLU activation functions. Thus the gradient of the interpretation loss $\ell_\mathrm{int}$ provides little information for updating $x$, which makes directly applying \meq{eq:update} ineffective.

\begin{figure}
		\centering
		\epsfig{file=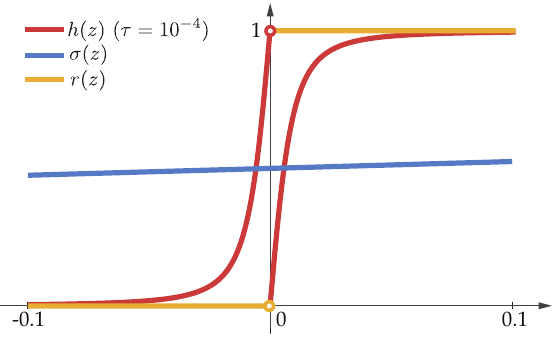, width=60mm}
		\caption{\small Comparison of $h(z)$, $\sigma(z)$, and $r(z)$ near $z = 0$.
		 \label{fig:gh}}
\end{figure}

 To overcome this, when performing back-propagation, we smooth the gradient of ReLU, denoted by $r(z)$, with a function $h(z)$ defined as ($\tau$ is a small constant, e.g., $\ssup{10}{-4}$):
\begin{equation*}
	h(z) \triangleq \begin{cases}
	(z + \sqrt{z^2 + \tau})' = 1 + z/\sqrt{z^2 + \tau} & (z < 0) \\
	(\sqrt{z^2 + \tau} )' = z/\sqrt{z^2 + \tau} & (z \geq 0)
	\end{cases}
\end{equation*}

Intuitively, $h(z)$ tightly approximates $r(z)$, while its gradient is non-zero everywhere.
Another possibility is the sigmoid function $\sigma(z) = 1/(1+e^{-z})$. Figure\mref{fig:gh} compares different functions near $z = 0$. Our evaluation shows that $h(z)$ significantly outperforms $\sigma(z)$ and $r(z)$ in attacking \grd.

This attack is extensible to other back-propagation-based interpreters (e.g., {\small{\sc DeepLIFT}}\mcite{Shrikumar:2017:icml}, {\small{\sc SmoothGrad}}\mcite{Smilkov:iclr:2017}, and {\small{\sc LRP}}\mcite{Bach:2015:plos}), due to their fundamentally equivalent, gradient-centric formulations\mcite{Ancona:iclr:2018}.

\subsection{Representation-Guided Interpretation}
\label{sec:cam}

This class of interpreters leverage the feature maps at intermediate layers of \dnns to generate attribution maps. We consider class activation mapping (\cam)\mcite{zhou:cam} as a representative interpreter of this class.


At a high level, \cam performs global average pooling\mcite{lin:nin} over the feature maps of the last convolutional layer, and uses the outputs as features for a linear layer with softmax activation to approximate the model predictions. Based on this connectivity structure, \cam computes the attribution maps by projecting the weights of the linear layer back to the convolutional feature maps.


Formally, let $a_k[i, j]$ denote the activation of the $k$-th channel of the last convolutional layer at the spatial position $(i, j)$. The output of global average pooling is defined as $A_k = \sum_{i, j} a_k[i, j]$. Further let $w_{k,c}$ be the weight of the connection between the $k$-th input and the $c$-th output of the linear layer. The input to the softmax function for a class $c$ with respect to a given input $x$ is approximated by:
\begin{equation}
z_c(x) \approx \sum_{k} w_{k,c}\, A_k = \sum_{i, j} \sum_k w_{k,c}\, a_k [i, j]
\end{equation}

The class activation map $m_c$ is then given by:
\begin{equation}
m_c[i,j] = \sum_k w_{k,c}\, a_k[i, j]
\end{equation}

Due to its use of deep representations at intermediate layers, \cam generates attribution maps of high visual quality and limited noise and artifacts\mcite{lin:nin}.

We instantiate $g$ with a \dnn that concatenates the part of $f$ up to its last convolutional layer and a linear layer parameterized by $\{w_{k,c}\}$. To attack \cam, we search for $\ax$ using a sequence of gradient descent updates as defined in \meq{eq:update}.
This attack can be readily extended to other representation-guided interpreters (e.g., {\gcam}\mcite{selvaraju:gradcam}), with details deferred to Appendix A1.

\subsection{Model-Guided Interpretation}

Instead of relying on deep representations at intermediate layers, model-guided methods train a meta-model to directly predict the attribution map for any given input in a single feed-forward pass.
 We consider {\rts}\mcite{Dabkowski:nips:2017} as a representative method in this category.

For a given input $x$ in a class $c$, \rts finds its attribution map $m$ by solving the following optimization problem:
\begin{equation}
\begin{array}{rl}
\min_m & \lambda_1 r_\mathrm{tv}(m) + \lambda_2 r_\mathrm{av}(m) - \log \left( f_c \left( \phi (x; m) \right) \right) \\
& + \lambda_3 f_c \left( \phi ( x; 1 - m ) \right)^{\lambda_4} \\
\mathrm{s.t.} & 0 \leq m \leq 1
\end{array}
\label{opt:rtsinit}
\end{equation}


Here $r_\mathrm{tv}(m)$ denotes the total variation of $m$, which reduces noise and artifacts in $m$; $r_\mathrm{av}(m)$ represents the average value of $m$, which minimizes the size of retained parts; $\phi(x;m)$ is the operator using $m$ as a mask to blend $x$ with random colors and Gaussian blur, which captures the impact of retained parts (where the mask is non-zero) on the model prediction; the hyper-parameters $\{\lambda_i\}_{i=1}^4$ balance these factors. Intuitively, this formulation finds the sufficient and necessary parts of $x$, based on which $f$ is able to make the prediction $f(x)$ with high confidence.

However, solving \meq{opt:rtsinit} for every input during inference is fairly expensive. Instead, \rts trains a \dnn to directly predict the attribution map for any given input, without accessing to the \dnn $f$ after training. In\mcite{Ronneberger:miccai:2015}, this is achieved by composing a {\rnet}\mcite{he:resnet} pre-trained on ImageNet\mcite{Deng:2009:cvpr} as the encoder (which extracts feature maps of given inputs at different scales) and a {\sc U-Net}\mcite{Ronneberger:miccai:2015} as the masking model, which is then trained to directly optimize \meq{opt:rtsinit}. We consider the composition of this encoder and this masking model as the interpreter $g$.

To attack \rts, one may directly apply \meq{eq:update}. However, our evaluation shows that this strategy is often ineffective for finding  desirable adversarial inputs. This is explained by that the encoder $\mathrm{enc}(\cdot)$ plays a significant role in generating attribution maps, while solely relying on the outputs of the masking model is insufficient to guide the attack. We thus add to \meq{eq:optatt2} an additional loss term $\ell_\mathrm{enc}(\mathrm{enc}(x), \mathrm{enc}(c_t))$, which measures the difference of the encoder's outputs for the adversarial input $x$ and the target class $\ay$.

We then search for the adversarial input $\ax$ with a sequence of gradient descent updates defined in \meq{eq:update}. More implementation details are discussed in \msec{sec:impl}.



\subsection{Perturbation-Guided Interpretation}

The fourth class of interpreters formulate finding the attribution map by perturbing the input with minimum noise and observing the change in the model prediction. We consider {\mask}\mcite{fong:mask} as a representative interpreter in this class.

For a given input $x$, \mask identifies its most informative parts by checking whether changing such parts influences the prediction $f(x)$. It learns a mask $m$, where $m[i] = 0$ if the $i$-th input feature is retained and $m[i] = 1$ if the feature is replaced with Gaussian noise. The optimal mask is found by solving an optimization problem:
\begin{align}
 \min_m  \,\,  f_c( \phi ( x; m ))  + \lambda \| 1 -m \|_1 \quad \mathrm{s.t.}  \,\,\,\, 0 \leq m \leq 1
\label{opt:maskopt4}
\end{align}
where $c$ denotes the current prediction $c = f(x)$ and $\phi(x; m)$ is the perturbation operator which blends $x$ with Gaussian noise. The first term finds $m$ that causes the probability of $c$ to decrease significantly, while the second term encourages $m$ to be sparse. Intuitively, solving \meq{opt:maskopt4} amounts to finding the most informative and necessary parts of $x$ with respect to its prediction $f(x)$. Note that this formulation may result in significant artifacts in $m$. A more refined formulation is given in Appendix A2.


%
%
%
%


Unlike other classes of interpreters, to attack \mask, it is infeasible to directly optimize \meq{eq:optatt2} with iterative gradient descent (\meq{eq:update}), because the interpreter $g$ itself is formulated as an optimization procedure.

Instead, we reformulate \system using a bilevel optimization framework. For given $\bx$, $\ay$, $\tm$, $f$, and $g$, we re-define the adversarial loss function as $\ell_\mathrm{adv}(x,m) \triangleq \ell_\mathrm{prd}(f(x), \ay) + \lambda \ell_\mathrm{int}(m, \tm)$ by introducing $m$ as an additional variable. Let $\ell_\mathrm{map}(m;x)$ be the objective function defined in \meq{opt:maskopt4} (or its variant \meq{opt:maskopt}). Note that $m_*(x) = \arg\min_m \ell_\mathrm{map}(m;x)$ is the attribution map found by \mask for a given input $x$. We then have the following attack framework:
\begin{align}
	\nonumber \min_x & \quad  \ell_\mathrm{adv}\left(x,\,m_*(x)\right) \\
	\mathrm{s.t.} &  \quad m_*(x) = \arg\min_m \,\ell_\mathrm{map}(m;x)
	\label{opt:maskopt2}
\end{align}

Still, solving the bilevel optimization in \meq{opt:maskopt2} exactly is challenging, as it requires recomputing $m_*(x)$ by solving the inner optimization problem whenever $x$ is updated. We propose an approximate iterative procedure which optimizes $x$ and $m$ by alternating between gradient descent on $\ell_\mathrm{adv}$ and $\ell_\mathrm{map}$ respectively.

More specifically, at the $i$-th iteration, given the current input $\ssup{x}{(i-1)}$, we compute its attribution map $m^{(i)}$ by updating $\ssup{m}{(i-1)}$ with gradient descent on $\ell_\mathrm{map}\left(\ssup{m}{(i-1)}; \ssup{x}{(i-1)}\right)$; we then fix $m^{(i)}$ and obtain $\ssup{x}{(i)}$ by minimizing $\ell_\mathrm{adv}$ after a single step of gradient descent with respect to $m^{(i)}$. Formally, we define the objective function for updating $x^{(i)}$ as:
\begin{displaymath}
\ell_\mathrm{adv}\left( \ssup{x}{(i-1)}, \,m^{(i)} - \xi \nabla_m \ell_\mathrm{map}\left(\ssup{m}{(i)}; \ssup{x}{(i-1)}\right)\right)
\end{displaymath}
where $\xi$ is the learning rate for this virtual gradient descent.

The rationale behind this procedure is as follows. While it is difficult to directly minimizing $\ell_\mathrm{adv}\left(x, m_*(x)\right)$ with respect to $x$, we use a single-step unrolled map as a surrogate of $m_*(x)$. A similar approach is used in\mcite{Finn:2017:icml}. Essentially, this iterative optimization defines a Stackelberg game\mcite{Scherer:stackelberg:1996} between the optimizer for $x$ (leader) and the optimizer for $m$ (follower), which requires the leader to anticipate the follower's next move to reach the equilibrium.

\begin{algorithm}[!ht]{\footnotesize
\KwIn{$\bx$: benign input; $\ay$: target class;  $\tm$: target map;  $f$: target DNN; $g$: MASK interpreter}
\KwOut{$\ax$: adversarial input}
initialize $x$ and $m$ as $\bx$ and $g(\bx;f)$\;
\While{not converged}{
\tcp{\footnotesize update $m$}
update $m$ by gradient descent along $\nabla_m \ell_\mathrm{map}(m;x)$\;
\tcp{\footnotesize update $x$ with single-step lookahead}
update $x$ by gradient descent along $\nabla_x \ell_\mathrm{adv}\left(x, \, m - \xi \nabla_m \ell_\mathrm{map}\left(m; x\right)\right)$\;
}
\Return $x$\;
\caption{\small \system against \mask. \label{alg:mask}}}
\end{algorithm}

Algorithm\mref{alg:mask} sketches the attack against \mask. More implementation details are given in \msec{sec:impl}. The theoretical justification for its effectiveness is deferred to Appendix A3.


%

%% file: implementation.tex
\subsection{Implementation and Optimization}
\label{sec:impl}

Next we detail the implementation of \system and present a suite of optimizations to improve the attack effectiveness against specific interpreters.

\vspace{2pt}
{\bf Iterative Optimizer --} We build the optimizer based upon {\pgd}\mcite{madry:towards}, which iteratively updates the adversarial input using \meq{eq:update}. By default, we use $\mathcal{L}_\infty$ norm to measure the perturbation magnitude.
It is possible to adopt alternative frameworks if other perturbation metrics are considered. For instance, instead of modifying pixels directly, one may generate adversarial inputs via spatial transformation\mcite{Xiao:2018:iclr,Alaifari:2019:iclr}, in which the perturbation magnitude is often measured by the overall spatial distortion. We detail and evaluate spatial transformation-based \system in \msec{sec:eval}.

\vspace{2pt}
{\bf Warm Start --} It is observed in our evaluation that it is often inefficient to search for adversarial inputs by running the update steps of \system (\meq{eq:update}) from scratch. Rather, first running a fixed number (e.g., 400) of update steps of the regular adversarial attack and then resuming the \system update steps significantly improves the search efficiency. Intuitively, this strategy first quickly approaches the manifold of adversarial inputs, and then searches for inputs  satisfying both prediction and interpretation constraints.

\vspace{2pt}
{\bf Label Smoothing --} Recall that we measure the prediction loss $\ell_\mathrm{prd}(f(x),c_t)$ with cross entropy. When attacking \grd, \system may generate intermediate inputs that cause $f$ to make over-confident predictions (e.g., with probability 1). The all-zero gradient of $\ell_\mathrm{prd}$ prevents the attack from finding inputs with desirable interpretations. To solve this, we refine cross entropy with label smoothing\mcite{Szegedy:2016:cvpr}. We sample $y_{c_t}$ from a uniform distribution $\mathbb{U}(1 - \rho, 1)$ and define $y_{c} = \frac{1-y_{c_t}}{|\mathcal{C}|-1}$ for $c \neq c_t$ and $\ell_\mathrm{prd}(f(x), c_t) = - \sum_{c \in \mathcal{C}} y_c \log f_c(x) $.
During the attack, we gradually decrease $\rho$ from 0.05 to 0.01.

%

%


\vspace{2pt}
{\bf Multistep Lookahead --} In implementing Algorithm\mref{alg:mask}, we apply multiple steps of gradient descent in both updating $m$ (line 3) and computing the surrogate map $m_*(x)$ (line 4), which is observed to lead to faster convergence in our empirical evaluation. Further, to improve the optimization stability, we use the average gradient to update $m$. Specifically, let $\{\sboth{m}{j}{(i)}\}$ be the sequence of maps obtained at the $i$-th iteration by applying multistep gradient descent. We use the aggregated interpretation loss $\sum_j\| \sboth{m}{j}{(i)} - \tm \|_2^2$ to compute the gradient for updating $m$.

\vspace{2pt}
{\bf Adaptive Learning Rate --} To improve the convergence of Algorithm\mref{alg:mask}, we also dynamically adjust the learning rate for updating $m$ and $x$. At each iteration, we use a running Adam optimizer as a meta-learner\mcite{Andrychowicz:nips:2016} to estimate the optimal learning rate for updating $m$ (line 3). We update $x$ in a two-step fashion to stabilize the training: (i) first updating $x$ in terms of the prediction loss $\ell_\mathrm{prd}$, and (ii) updating it in terms of the interpretation loss $\ell_\mathrm{int}$.
During (ii), we use a
binary search to find the largest step size, such that $x$'s confidence is still above a certain threshold $\kappa$ after the perturbation. 

{\bf Periodical Reset --} Recall that in Algorithm\mref{alg:mask}, we update the estimate of the attribution map by following gradient descent on $\ell_\mathrm{map}$. As the number of update steps increases, this estimate may deviate significantly from the true map generated by the \mask interpreter, which negatively impacts the attack effectiveness. To address this, periodically (e.g., every 50 iterations), we replace the estimated map with the map $g(x;f)$ that is directly computed by \mask based on the current adversarial input. At the same time, we reset the Adam step parameter to correct its internal state.

%% file: evaluation.tex
\section{Attack Evaluation}
\label{sec:eval}

Next we conduct an empirical study of \system on a variety of \dnns and interpreters from both qualitative and quantitative perspectives. Specifically, our experiments are designed to answer the following key questions about \system:

$\mathrm{Q_1}$: Is it effective to deceive target classifiers?

$\mathrm{Q_2}$: Is it effective to mislead target interpreters?

$\mathrm{Q_3}$: Is it evasive with respect to attack detection methods?

$\mathrm{Q_4}$: Is it effective in real security-critical applications?

$\mathrm{Q_5}$: Is it flexible to adopt alternative attack frameworks?


%

\subsection*{Experimental Setting}
\label{sec:evalimpl}

We first introduce the setting of our empirical evaluation.

{\bf Datasets --} Our evaluation primarily uses ImageNet\mcite{Deng:2009:cvpr}, which consists of 1.2 million images from 1,000 classes. Every image is center-cropped to 224$\times$224 pixels. For a given classifier $f$, from the validation set of ImageNet, we randomly sample 1,000 images that are classified correctly by $f$ to form our test set. All the pixels are normalized to $[0, 1]$. 

{\bf Classifiers --} We use two state-of-the-art \dnns as the classifiers, {\rnet-50}\mcite{he:resnet} and {\dnet-169}\mcite{Huang:cvpr:2017}, which respectively attain 77.15\% and 77.92\% top-1 accuracy on ImageNet. Using two \dnns of distinct capacities (50 layers versus 169 layers) and architectures (residual blocks versus dense blocks), we factor out the influence of the characteristics of individual \dnns.

{\bf Interpreters --} We adopt {\grd}\mcite{Simonyan:gradsaliency}, {\cam}\mcite{zhou:cam}, {\rts}\mcite{Dabkowski:nips:2017}, and {\mask}\mcite{fong:mask} as the representatives of back-propagation-, representation-, model-, and perturbation-guided interpreters respectively. We adopt their open-source implementation in our evaluation.
 As \rts is tightly coupled with its target \dnn (i.e., \rnet), we train a new masking model for \dnet. To assess the validity of the implementation, we evaluate all the interpreters in a weakly semi-supervised localization task\mcite{Cao:2015:iccv} using the benchmark dataset and method in \mcite{Dabkowski:nips:2017}. Table\mref{tab:weaklyError} summarizes the results. The performance of all the interpreters is consistent with that reported in\mcite{Dabkowski:nips:2017}, with slight variation due to the difference of underlying \dnns.

 %

 \begin{table}[!ht]{\footnotesize
	\centering
	\begin{tabular}{c|c|c|c|c}
	Interpreter	& \grd & \cam& \mask & \rts \\
		\hline
		\hline

	Measure & $43.1$ &$43.8$ & $45.2$&  $34.2$\\

	\end{tabular}
	\caption{\small Performance of the interpreters in this paper in a weakly semi-supervised localization task (with \rnet as the classifier). \label{tab:weaklyError}}}
\end{table}


{\bf Attacks --} We implement all the variants of \system in \msec{sec:attack} on the \pgd framework. In addition, we also implement \system on a spatial transformation framework (\adef)\mcite{Alaifari:2019:iclr}. We compare \system with regular {\pgd}\mcite{madry:towards}, a universal first-order adversarial attack. For both \system and \pgd, we assume the setting of targeted attacks, in which the adversary attempts to force the \dnns to misclassify the adversarial inputs into randomly designated classes. The parameter settings of all the attacks are summarized in Appendix B.

\subsection*{Q1. Attack Effectiveness (Prediction)}

We first evaluate the effectiveness of \system in terms of deceiving target \dnns. The effectiveness is measured using {\em attack success rate}, which is defined as
\begin{displaymath}
\mathrm{Attack\,\,Success\,\,Rate}\,(\mathrm{ASR}) = \frac{\mathrm{\#\, successful\,\,trials}}{\mathrm{\#\, total\,\,trials}}
\end{displaymath}
and {\em misclassification confidence} (MC), which is the probability assigned by the \dnn to the target class $c_t$.

\begin{table}[!ht]{\footnotesize
	\centering
  \setlength\tabcolsep{3pt}
	\begin{tabular}{c|cccc|cccc}
		& \multicolumn{4}{c|}{\rnet} &\multicolumn{4}{c}{\dnet}\\ \cline{2-9}
		& {\grd} & {\cam} & {\mask} & {\rts} & {\grd} & {\cam} & {\mask} & {\rts} \\
		\hline
		\hline
		\multirow{1}{*}{\texttt{P}}    &          \multicolumn{4}{c|}{100\% (1.0)}         &       \multicolumn{4}{c}{100\% (1.0)}      \\

		\rowcolor{Gray}
		& 100\%  &  100\% & 98\% & 100\%  & 100\% & 100\% & 96\%  & 100\% \\
		\rowcolor{Gray}
		\multirow{-2}{*}{\texttt{A}}  & (0.99) & (1.0) & (0.99) & (1.0) & (0.98) & (1.0) & (0.98)  & (1.0)\\
		\hline
	\end{tabular}
	\caption{\small Effectiveness of \pgd (P) and \system (A) against different classifiers and interpreters in terms of ASR (MC).
		\label{tab:attacksucc}}
    }
\end{table}

Table\mref{tab:attacksucc} summarizes the attack success rate and misclassification confidence of \system and \pgd against different combinations of classifiers and interpreters. Note that as \pgd is only applied on the classifier, its effectiveness is agnostic to the interpreters. To make fair comparison, we fix the maximum number of iterations as 1,000 for both attacks. It is observed that \system achieves high success rate (above 95\%) and misclassification confidence (above 0.98) across all the cases, which is comparable with the regular \pgd attack. We thus have the following conclusion.
\begin{tcolorbox}[boxrule=0pt, title= Observation 1]
Despite its dual objectives, \system is as effective as regular adversarial attacks in deceiving target \dnns.
\end{tcolorbox}

%

\subsection*{Q2. Attack Effectiveness (Interpretation)}

Next we evaluate the effectiveness of \system in terms of generating similar interpretations to benign inputs. Specifically, we compare the interpretations of benign and adversarial inputs, which is crucial for understanding the security implications of using interpretability as a means of defenses\mcite{Du:2018:kdd,Tao:nips:2018}. Due to the lack of standard metrics for interpretation plausibility, we use a variety of measures in our evaluation.

\begin{figure}[!ht]
	\hspace{-10pt}
	\epsfig{file = 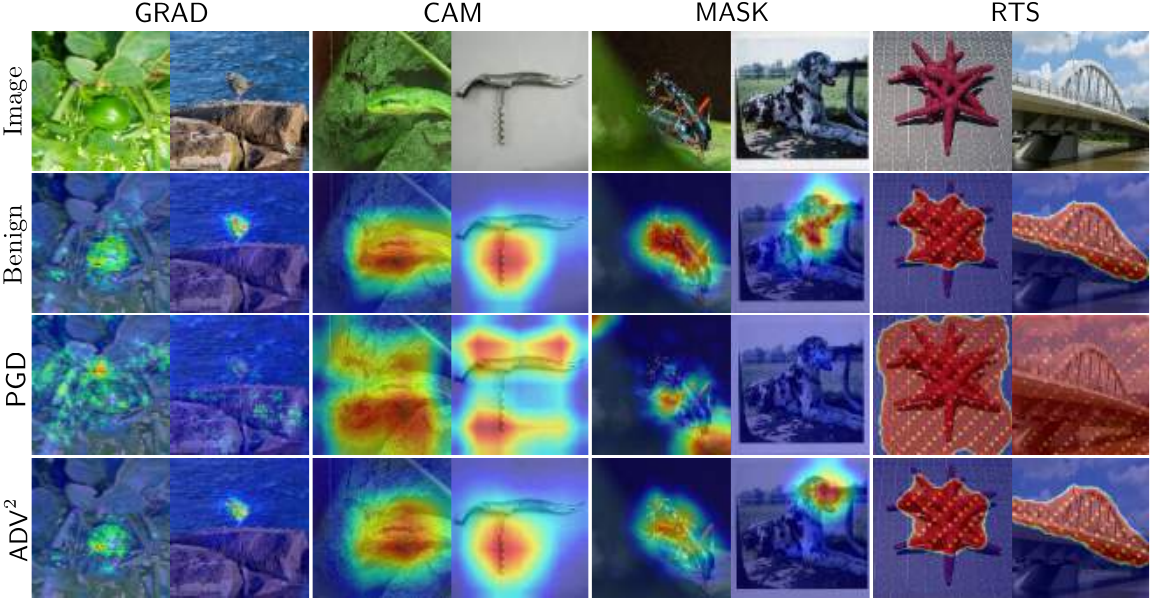, width=90mm}
	\caption{\small Attribution maps of benign and adversarial (\pgd, \system) inputs with respect to \grd, \cam, \mask, and \rts on \rnet. 	\label{fig:ensemble-sample}}
\end{figure}
{\bf Visualization --} We first qualitatively compare the interpretations of benign and adversarial (\pgd, \system) inputs. 
Figure\mref{fig:ensemble-sample} show a set of sample inputs and their attribution maps with respect to \grd, \cam, \mask, and \rts (more samples in Appendix C1). Observe that in all the cases, the \system inputs generate interpretations perceptually indistinguishable from their benign counterparts. In comparison, the \pgd inputs are easily identifiable by inspecting their attribution maps.

\vspace{2pt}
{\bf $\bm{L_p}$ Measure --}
Besides qualitatively comparing the attribution maps of benign and adversarial inputs, we also measure their similarity quantitatively. By considering attribution maps as matrices, we measure the $\mathcal{L}_1$ distance between benign and adversarial maps. Figure\mref{tab:lpdistance} summarizes the results (other $\mathcal{L}_\mathrm{p}$ measures in Appendix C1). For comparison, we normalize all the measures to $[0, 1]$ by dividing them by the number of pixels.

\begin{figure}[!ht]
	\epsfig{file = 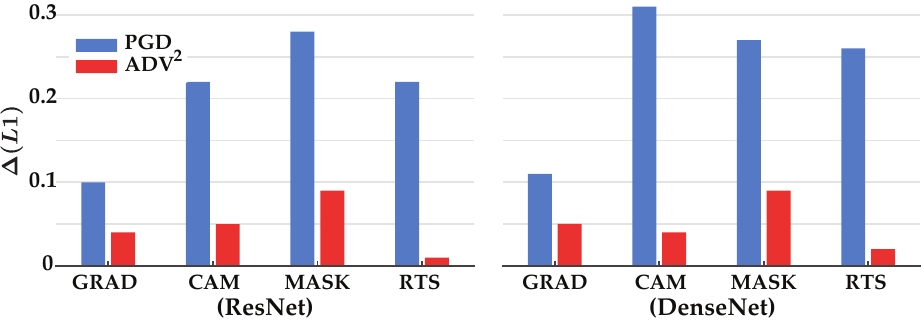, width=85mm}
	\caption{\small Average $\mathcal{L}_1$ distance between benign and adversarial (\pgd, \system) attribution maps. \label{tab:lpdistance}
  }
\end{figure}

We have the following observations. (i) Compared with \pgd, \system generates attribution maps much more similar to benign cases. The average $\mathcal{L}_1$ measure of \system is more than 60\% lower than \pgd across all the interpreters. (ii) The effectiveness of \system varies with the target interpreter. For instance, compared with other interpreters, the difference between \pgd and \system is relatively marginal on \grd, implying that different interpreters may inherently feature varying robustness against \system. (iii) The effectiveness of \system seems insensitive to the underlying \dnn. On both \rnet and \dnet, it achieves similar $\mathcal{L}_1$ measures.

\vspace{2pt}
{\bf IoU Test --}
Another quantitative measure for the similarity of attribution maps is the intersection-over-union (\iou) score. It is widely used in object detection\mcite{He:iccv:2017} to compare model predictions with ground-truth bounding boxes. Formally, the \iou score of a binary-valued map $m$ with respect to a baseline map $\bbm$ is defined as their Jaccard similarity:
$\mathrm{IoU}(m) = |O(m) \cap O(\bbm)|/|O(m) \cup O(\bbm)|$,
where $O(m)$ denotes the set of non-zero dimensions in $m$. In our case, as the values of attribution maps are floating numbers, we first apply threshold binarization on the maps.

\begin{figure}[!ht]
	\epsfig{file = 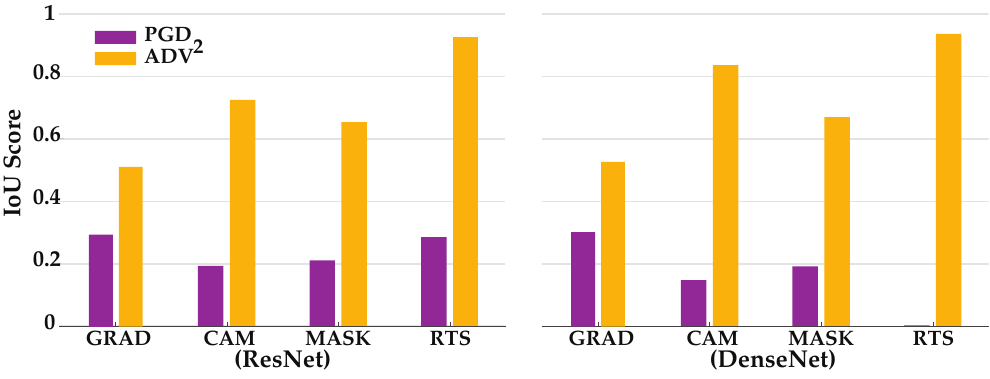, width=85mm}
	\caption{\small IoU scores of adversarial attribution maps (\pgd, \system) with respect to benign maps.  \label{fig:ioutest}
  }
\end{figure}

Following a typical rule used in the object detection task\mcite{He:iccv:2017} where a detected region of interest (RoI) is considered positive if its \iou score is above 0.5 with respect to a ground-truth mask, we thus consider an attribution map as plausible if its \iou score exceeds 0.5 with respect to the benign attribution map.
Figure\mref{fig:ioutest} compares the average \iou scores of adversarial maps (\pgd, \system) with respect to the benign cases. Observe that \system achieves \iou scores above $0.5$ across all the interpreters, which are more than $40\%$ higher than \pgd in all the cases. Especially on \rts, in which the attribution maps are natively binary-valued, \system achieves \iou scores above 0.9 on both \rnet and \dnet.

Based on both qualitative and quantitative measures, we have the following conclusion.
\begin{tcolorbox}[boxrule=0pt, title= Observation 2]
\system is able to generate adversarial inputs with interpretations highly similar to benign cases.
\end{tcolorbox}

\subsection*{Q3. Attack Evasiveness}

Intuitively, from the adversary's perspective, \system entails a search space for adversarial inputs no larger than its underlying adversarial attack (e.g., \pgd), as \system needs to optimize both the prediction loss $\ell_\mathrm{prd}$ and interpretation loss $\ell_\mathrm{int}$, while \system only needs to optimize $\ell_\mathrm{prd}$. Next we compare \pgd and \system in terms of their evasiveness with respect to adversarial attack detection methods.




\vspace{2pt}
{\bf Basic ADV$^2$ --}
To be succinct, we consider feature squeezing (\fs)\mcite{Xu:2018:ndss} as a concrete detection method. \fs reduces the adversary's search space by coalescing inputs corresponding to different feature vectors into a single input, and detects adversarial inputs by comparing their predictions under original and squeezed settings. This operation is implemented in the form of a set of ``squeezers'': bit depth reduction, local smoothing, and non-local smoothing.

\begin{table}[!ht]{\footnotesize
	\centering
    \setlength\tabcolsep{3pt}
	\begin{tabular}{c|c|c|cc|cc}
		Squeezer  & Setting  & PGD & MASK-\texttt{A} & RTS-\texttt{A} & \cellcolor{Gray} MASK-\texttt{A}$^*$ & \cellcolor{Gray} RTS-\texttt{A}$^*$  \\
		\hline
		\hline
		Bit Depth 	& 2-bit        &92.3\%   &84.1\% &94.0\%   &\cellcolor{Gray} 11.7\% \cellcolor{Gray} &\cellcolor{Gray} 29.4\% \\
	Reduction		& 3-bit        &72.7\%   &89.2\% &88.3\%   &\cellcolor{Gray} 35.9\%  \cellcolor{Gray} &\cellcolor{Gray} 13.9\% \\
		\hline
		L. Smoothing
		& 3$\times$3   &97.3\%   &98.6\% &99.0\%  &\cellcolor{Gray} 16.5\%&\cellcolor{Gray} 3.4\% \\
		\hline
	N. Smoothing
		& 11-3-4& 52.3\%    &74.7\% &75.3\% & \cellcolor{Gray} 51.7\% &\cellcolor{Gray} 29.4\% \\
		\hline

	\end{tabular}
	\caption{\small Detectability of adversarial inputs by \pgd, basic \system (A), and adaptive \system (A$^*$) using feature squeezing. 	\label{tab:squeezing}}}
\end{table}

Table\mref{tab:squeezing} lists the detection rate of adversarial inputs (\pgd, \system) using different types of squeezers on \rnet. Observe that the squeezers seem effective to detect both \system and \pgd inputs. For instance, local smoothing achieves higher than 97\% success rate in detecting both \system and \pgd inputs, with difference less than 2\%. We thus have:

\begin{tcolorbox}[boxrule=0pt, title= Observation 3]
The overall detectability of \system and \pgd with respect to feature squeezing is not significantly different.
\end{tcolorbox}

{\bf Adaptive ADV$^2$ --}
We now adapt \system to evade the detection of \fs. Related to existing adaptive attacks against {\fs}\mcite{He:usenixworkshop:2017}, this optimization is interesting in its own right. Specifically, for smoothing squeezers, we augment the loss function $\ell_\mathrm{adv}(x)$ (\meq{eq:optatt2}) with the term $\ell_\mathrm{sqz}(f(x), f(\psi(x))$, which is the cross entropy of the predictions of original and squeezed inputs ($\psi$ is the squeezer).

\begin{algorithm}{\small
\KwIn{$\bx$: benign input; $\ay$: target class; $f$: target DNN; $g$: target interpreter; $\psi$: bit depth reduction; $i$: bit depth}
\KwOut{$\ax$: adversarial input}
\tcp{\footnotesize augmented $\ell_\mathrm{adv}$ with $\ell_\mathrm{sqz}$ $\mathtt{w.r.t.}$ smoothing}
\tcp{\footnotesize attack in squeezed space}
$\ssub{x}{+} \leftarrow$ PGD on $\psi(\bx)$ with target $\ay$ and $\alpha = 1/2^i$\;
\tcp{\footnotesize attack in original space}
search for $\ax = \arg\min_{x \in \mathcal{B}_\epsilon(\bx)} \ell_\mathrm{adv}(x) + \lambda \| f(x) - f(\ssub{x}{+}) \|_1$\;
\Return $\ax$\;
\caption{\small Adaptive \system against Feature Squeezing.\label{alg:adaptive}}}
\end{algorithm}

For bit depth reduction, we use a two-stage strategy. (i) We first search in the squeezed space for an adversarial input $\ssub{x}{+}$ that is close to $\bx$'s $\epsilon$-neighborhood. To do so, we run \pgd over $\psi(\bx)$ with learning rate $\alpha = 1/2^i$ ($i$ is the bit depth). (ii) We then search in $\bx$'s $\epsilon$-neighborhood for an adversarial input $\ax$ that is classified similarly as $\ssub{x}{+}$. To do so, we augment the loss function $\ell_\mathrm{adv}(x)$ with a probability loss term
$\|f(x) - f(\ssub{x}{+}) \|_1$ ($f(\ssub{x}{+})$ is $\ssub{x}{+}$'s probability vector), and then apply \pgd to search for $\ax$ within $\bx$'s $\epsilon$-neighborhood. The overall algorithm is sketched in Algorithm\mref{alg:adaptive}.

\begin{table}[!ht]{\footnotesize
	\centering
 \setlength\tabcolsep{5pt}
	\begin{tabular}{c|ccc|ccc}
  	\multirow{2}{*}{Metric}    &  \multicolumn{3}{c|}{MASK} &   \multicolumn{3}{c}{RTS}\\
    \cline{2-7}
 & \texttt{P} & \texttt{A} &\cellcolor{Gray}  \texttt{A}$^*$ & \texttt{P} & \texttt{A} &\cellcolor{Gray}  \texttt{A}$^*$ \\
		\hline
		\hline
$\Delta \mathcal{L}_1$ &0.28 & 0.09 &\cellcolor{Gray}  0.09 & 0.22 & 0.01 & \cellcolor{Gray} 0.02 \\
\hline
IoU & 0.21& 0.65 &\cellcolor{Gray}  0.61 & 0.29 & 0.93 &\cellcolor{Gray}  0.94\\
	\end{tabular}
	\caption{\small $\mathcal{L}_1$ measures and \iou scores of adversarial attribution maps (\pgd, basic and adaptive \system) with respect to benign maps.\label{tab:evasiveness}}}
\end{table}

Table\mref{tab:squeezing} summarizes the detection rate of adversarial inputs generated by adaptive \system, which drops significantly, compared with the case of basic \system. Note that here we only show the possibility of adapting \system to evade a representative detection method, and consider an in-depth study on this matter as our ongoing work.
Meanwhile, we compare the $\mathcal{L}_1$ measures and \iou scores of the attribution maps generated by basic and adaptive \system (with respect to the benign maps). Table\mref{tab:evasiveness} shows the results. Observe that the optimization in adaptive \system has little impact on its attack effectiveness against the interpreters. We may thus conclude:
\begin{tcolorbox}[boxrule=0pt, title= Observation 4]
It is possible to adapt \system to generate adversarial inputs evasive with respect to feature squeezing.
\end{tcolorbox}

\begin{figure}[!ht]
\hspace{-10pt}
\epsfig{file = 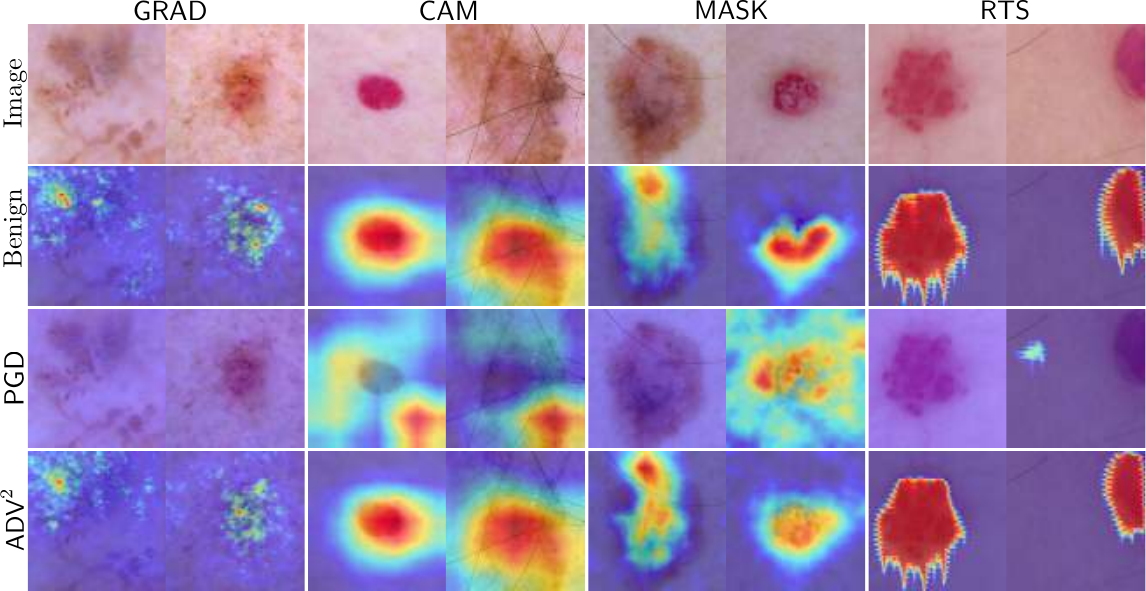, width=88mm}
\caption{\small Attribution maps of benign and adversarial (\system) inputs in the skin cancer screening application.
	\label{fig:skincancer}}
\end{figure}

\subsection*{Q4. Real Application}

We now evaluate the effectiveness of \system in real security-critical applications.
We use the skin cancer screening task from the ISIC 2018 challenge\mcite{gessert:2018:arxiv} as a case study, in which given skin lesion images are categorized into a seven-disease taxonomy. We adopt a competition-winning model\footnote{\url{https://github.com/ngessert/isic2018}} (with \rnet as its backbone) as the classifier,
which attains 82.27\% weighted multi-class accuracy on the holdout set (more details in Appendix C2).

We apply \system on this classifier and measure its effectiveness of generating plausible interpretations. Figure\mref{fig:skincancer} shows a set of samples and their attribution maps on the four interpreters. Observe that \system generates interpretations visually indiscernible from their benign counterparts in all the cases.

\begin{figure}[!ht]
	\centering
	\epsfig{file = 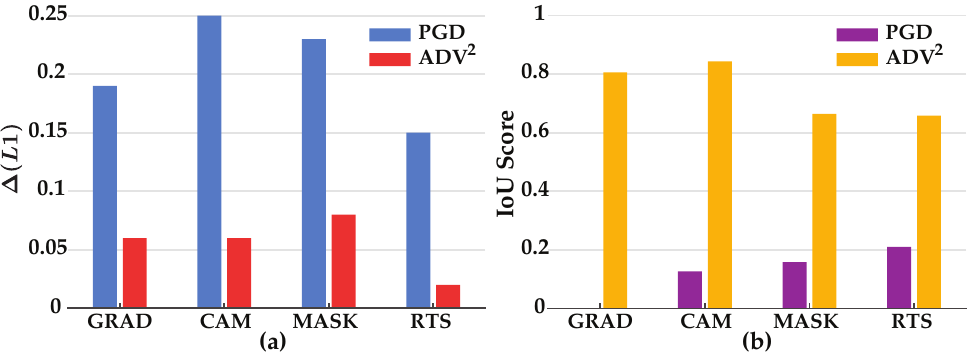, width=85mm}
	\caption{\small $\mathcal{L}_1$ measures (a) and IoU scores (b) of adversarial attribution maps (\pgd, \system) with respect to benign maps. \label{tab:iouSkinCancer}}
\end{figure}

This similarity is further quantitatively validated in Figure\mref{tab:iouSkinCancer}, which shows the $\mathcal{L}_1$ measures (other $\mathcal{L}_\mathrm{p}$ measures in Appendix C2)  and \iou scores of the maps generated by \system with respect to the benign maps. For instance, the \iou scores of \system exceed $0.62$ across all the interpreters.

%



\subsection*{Q5. Alternative Attack Framework}

Besides the \pgd framework, \system can also be flexibly built upon alternative frameworks.
Here we construct \system upon {\adef}\mcite{Xiao:2018:iclr}, a spatial transformation-based adversarial attack. 
The implementation details are given in Appendix A4.

\begin{figure}[!ht]
\hspace{-10pt}
		\epsfig{file = 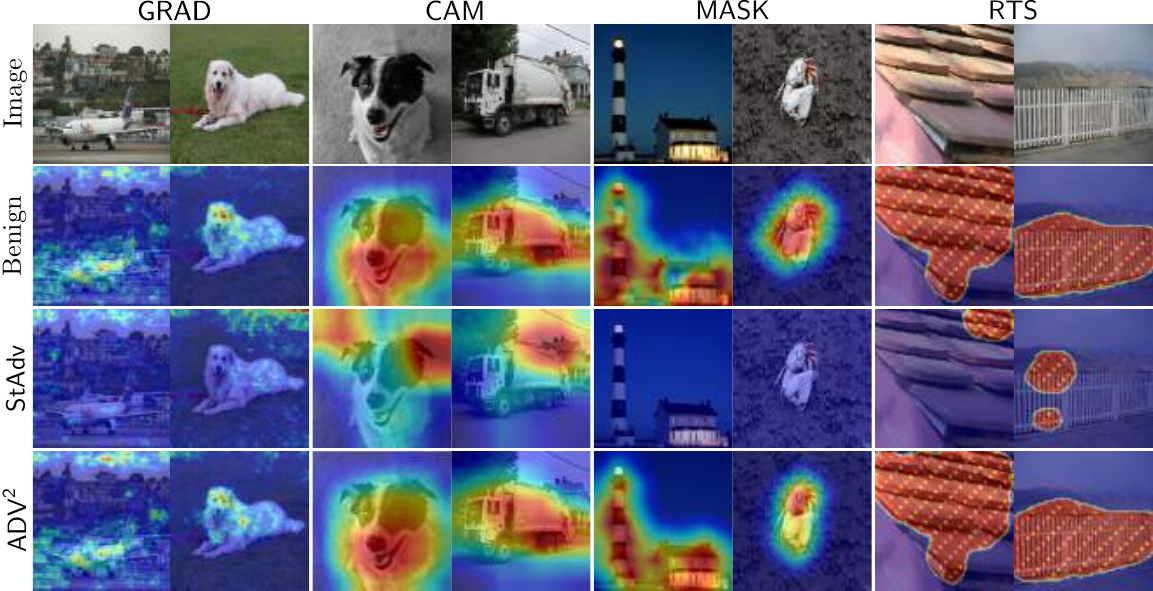, width=88mm}
		\caption{\small Attribution maps of benign and adversarial (\adef, \adef-based \system) inputs with respect to \grd, \cam, \mask, and \rts on \rnet.
			\label{fig:stimg}}
	\end{figure}

Figure\mref{fig:stimg} illustrates sample benign and adversarial (\adef, \system) inputs and their interpretations. Compared with \adef, \system generates adversarial inputs with maps much more similar to the benign cases, which indicates the effectiveness of \system constructed upon the \adef framework.

This observation is also confirmed by the $\mathcal{L}_1$ measures and \iou scores of adversarial attribution maps, which are shown in Figure\mref{fig:l1-regular} (more results in Appendix C3). Interestingly, compared with the other interpreters, \mask seems more resilient to \adef-based \system. The comparison with the results of \pgd-based \system (Figure\mref{tab:lpdistance} and \mref{fig:ioutest}) implies that the (relative) robustness of different interpreters may vary with the concrete attacks (details in \msec{sec:tradeoff}).

\begin{figure}[!ht]
  \epsfig{file = 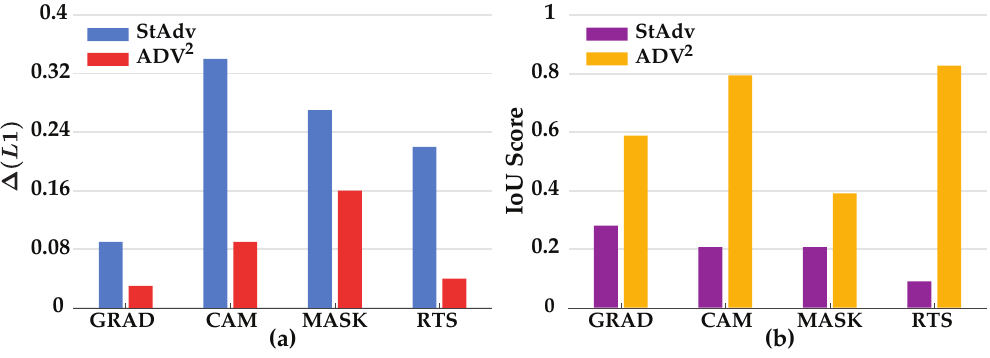, width=85mm}
  \caption{\small $\mathcal{L}_1$ measures (a) and \iou scores (b) of adversarial attribution maps (\adef, \adef-based \system) with respect to benign maps on \rnet. \label{fig:l1-regular} }
\end{figure}

Overall we have the following conclusion.
\begin{tcolorbox}[boxrule=0pt, title= Observation 5]
As a general class of attacks, \system can be flexibly built upon alternative adversarial attack frameworks.
\end{tcolorbox}




%% file: discussion.tex
\section{Discussion}
\label{sec:tradeoff}

While it is shown in \msec{sec:eval} that \system is effective against a range of classifiers and interpreters, the cause of this effectiveness is unclear yet. Next we conduct a study on this root cause from both analytical and empirical perspectives. Based on our findings, we further discuss potential countermeasures against \system.

%
%
%
%


\subsection*{Q1. Root of Attack Vulnerability}

Recall that the formulation of \system in \meq{eq:optatt2} defines two seemingly conflicting objectives: (i) maximizing the prediction change while (ii) minimizing the interpretation change. We thus conjecture that the effectiveness of \system may stem from the partial independence between a classifier and its interpreter -- the interpreter's explanations only partially describe the classifier's predictions, making it practical to exploit both models simultaneously.

To validate the existence of this prediction-interpretation gap, we consider \system targeting randomly generated predictions and interpretations. For a given input $\bx$, we randomly generate a target class $\ay$ and a target interpretation $m_t$, and search for an adversarial input $\ax$ that triggers the classifier to misclassify it as $\ay$ and also generates an interpretation similar to $m_t$ (i.e., $f(\ax) = \ay$ and $g(\ax;f) \approx \tm$). Intuitively, if \system is able to find such $\ax$, it indicates that the classifier and its interpreter can be manipulated separately; in other words, they are only partially aligned with each other.

\vspace{2pt}
{\bf Random Patch Interpretation --} In the first case,
 for a given input, we define its target attribution map by (i) sampling a patch of random shape (either a rectangle or a circle), random angle, and random position over the input, and (ii) setting the elements inside the patch as `1' and that outside it as `0'. Typically this target map deviates significantly from its benign counterpart, due to its randomness.
 %

 \begin{table}[!ht]{\footnotesize
 \centering
 \begin{tabular}{c|cccc}
  & \grd &	\cam	& \mask & \rts \\
 	 \hline
 		 \hline

\multirow{2}{*}{\system} & 100\% & 100\% & 99\% & 100\%  \\
 & (0.98) &  (1.0)    & (0.95) & (1.0)\\
 \end{tabular}
 \caption{\small {\small ASR} ({\small MC}) of \system targeting random patch interpretations. \label{tab:target}}}
\end{table}

We evaluate the effectiveness of \system under this setting. Table\mref{tab:target} summarizes the attack success rate of \system on \rnet. Observe that compared with Table\mref{tab:attacksucc},
targeting random patch interpretations has little impact on the attack effectiveness in terms of deceiving the classifiers, implying that the space of adversarial inputs is sufficiently large to contain ones with targeted interpretations.


\begin{figure}[!ht]
\hspace{-10pt}
	\epsfig{file = 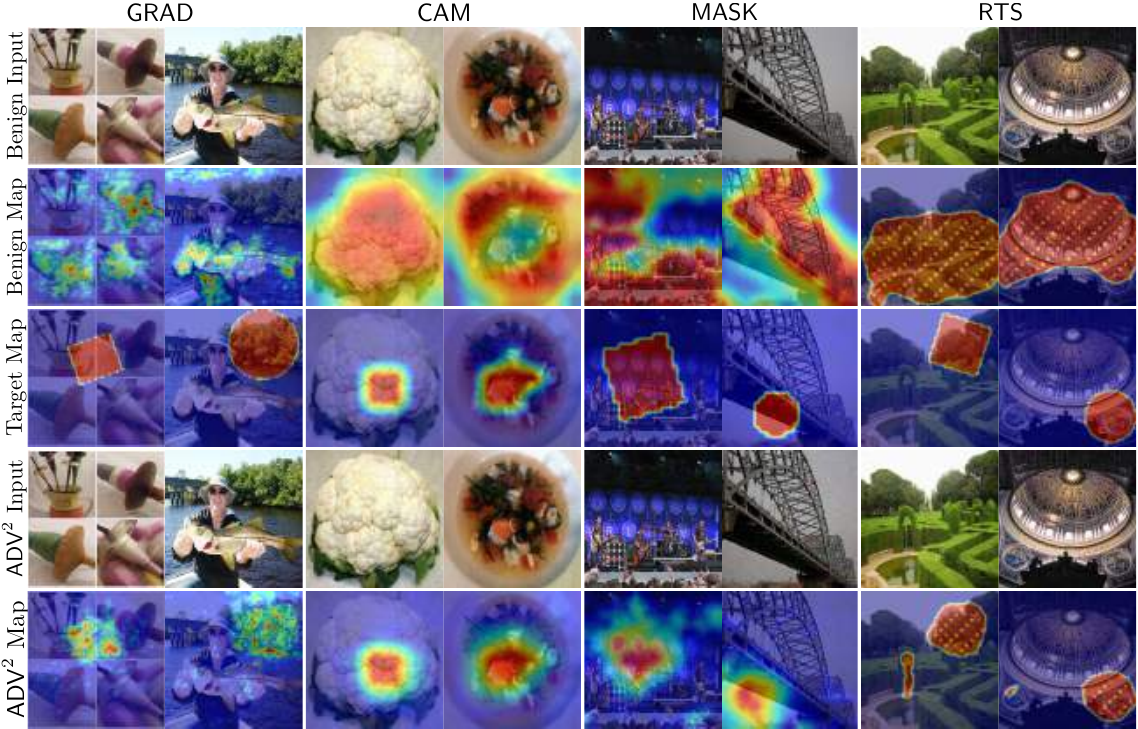, width=88mm}
	\caption{\small Visualization of \system targeting random patch interpretations across different interpreters on \rnet. 	\label{fig:target}}
\end{figure}

We then evaluate the effectiveness of \system in terms of generating the target interpretations. For a given benign input $\bx$ and a target random patch map $\tm$, \system attempts to generate an adversarial input $\ay$ with the interpretation similar to $\tm$. Figure\mref{fig:target} visualizes a set of sample results. Note that in all the cases the \system maps appear visually similar to the target maps, highlighting the attack effectiveness. This effectiveness is further validated in Table\mref{tab:targetdist}. Observe that across all the interpreters, an \system map is much more similar to its target map, compared with its benign counterpart.


\begin{table}[!ht]{\footnotesize
	\centering
	\begin{tabular}{c|cccc}

	& \grd & \cam & \mask & \rts\\
		\hline
		\hline
		 $\Delta_\mathrm{b}$$\mathcal{L}_1$ & $0.16$ & $0.50 $ & $0.42 $  & $0.49$ \\
		\cellcolor{Gray} $\Delta_\mathrm{t}$$\mathcal{L}_1$	& \cellcolor{Gray}$0.10$ & \cellcolor{Gray}$0.04$ &  \cellcolor{Gray}$0.15 $ & \cellcolor{Gray}$0.07$  \\
	\end{tabular}
	\caption{\small Comparison of \system and target maps ($\Delta_\mathrm{t}$) and that of \system and benign maps ($\Delta_\mathrm{b}$), measured by $\mathcal{L}_1$ distance. \label{tab:targetdist}}}
\end{table}

%

{\bf Random Class Interpretation --} In the second case, for a given input (with $\ay$ as the target class), we instantiate its target interpretation with the attribution map of a benign input randomly sampled from another class $\tilde{c}_t$. We particularly enforce $\ay \neq \tilde{c}_t$; in other words, the adversarial input is misclassified into one class but interpreted as another one.

\begin{table}[!ht]{\footnotesize
\centering
\begin{tabular}{c|cccc}
 & \grd &	\cam	& \mask & \rts \\
	\hline
		\hline

\multirow{2}{*}{\system} & 100\% & 100\% & 100\% & 100\%  \\
& (0.99) &  (0.99)    & (0.99) & (1.0)\\
\end{tabular}
\caption{\small {\small ASR} ({\small MC}) of \system with random class interpretations. \label{tab:succ-target-map}}}
\end{table}

The {\small ASR} of \system is summarized in Table\mref{tab:succ-target-map}. Observe that targeting random class interpretations has little influence on the attack effectiveness of deceiving the classifiers. Figure\mref{fig:target-map-res} visualizes a set of sample target and \system inputs and their interpretations (\dnet results in Appendix C4). Note that the target and \system inputs are fairly distinct, but with highly similar interpretations. This is quantitatively validated by their $\mathcal{L}_1$ measures and \iou scores listed in Figure\mref{fig:l1-target}.

\begin{figure}[!ht]
\hspace{-10pt}
	\epsfig{file = 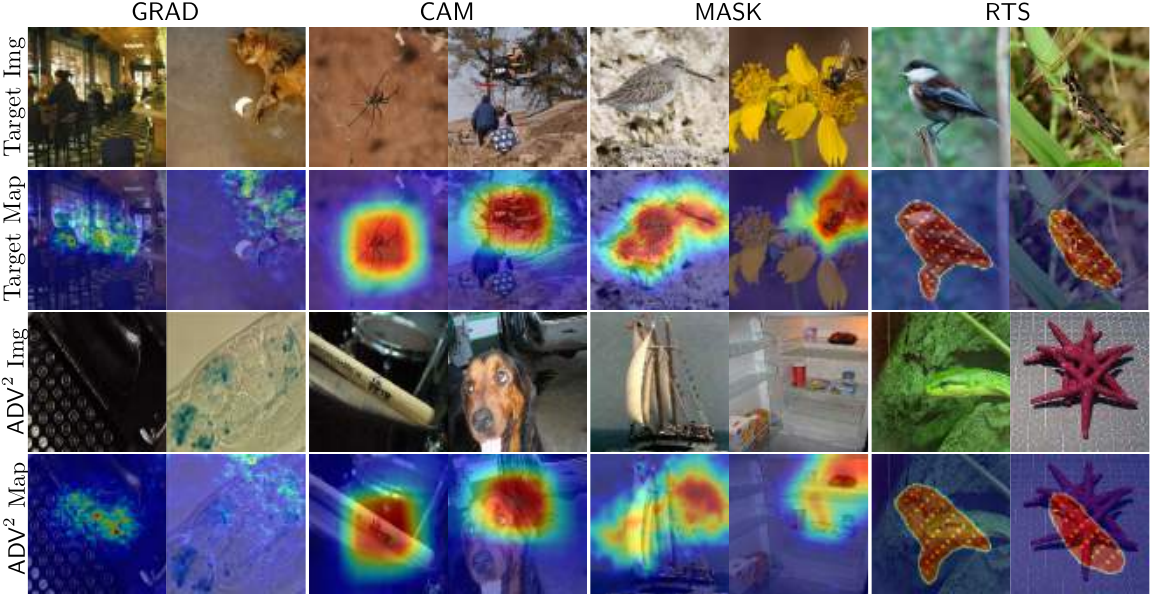, width=88mm}
	\caption{\small Target and adversarial (\system) inputs and their attribution maps on \rnet.
\label{fig:target-map-res}}
\end{figure}

\begin{figure}[!ht]
	\epsfig{file = 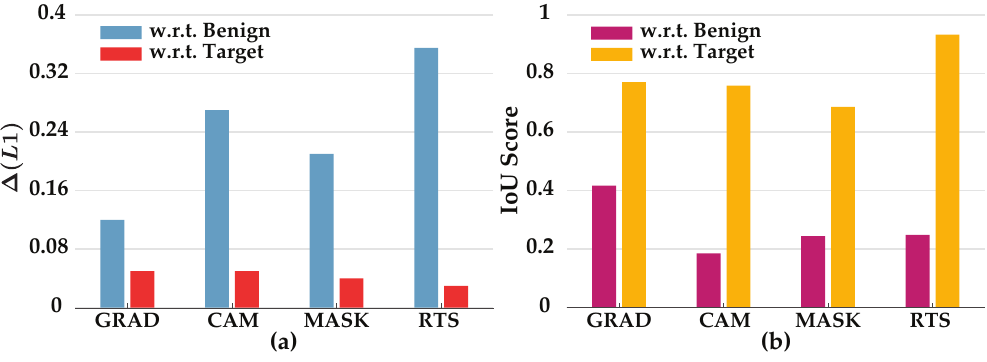, width=85mm}
	\caption{\small $\mathcal{L}_1$ measures (a) and IoU scores (b) of adversarial maps with respect to benign and target cases on \rnet.
		\label{fig:l1-target}}
\end{figure}



The experiments above show that it is practical to generate adversarial inputs targeting arbitrary predictions and interpretations. We can therefore conclude:
\begin{tcolorbox}[boxrule=0pt, title= Observation 6]
A \dnn and its interpreter are often not fully aligned, allowing the adversary to exploit both models simultaneously.
\end{tcolorbox}

\subsection*{Q2. Root of Prediction-Interpretation Gap}

Next we  explore the fundamental causes of this prediction-interpretation gap. We speculate one following possible explanation as: existing interpretation models do not comprehensively capture the dynamics of \dnns, each only describing one aspect of their behavior.

Specifically, \grd solely relies on the gradient information; \mask focuses on the input-prediction correspondence while ignoring the internal representations; \cam leverages the deep representations at intermediate layers, but neglecting the input-prediction correspondence; \rts uses the internal representations in an auxiliary encoder and the input-interpretation correspondence in the training data, which however may deviate from the true behavior of \dnns.

Intuitively the exclusive focus on one aspect (e.g., input-prediction correspondence) of the \dnn behavior results in loose constraints: when performing the attack, the adversary only needs to ensure that benign and adversarial inputs cause \dnns to behave similarly from one specific perspective. We validate this speculation from two observations, low attack transferability and disparate attack robustness.

\vspace{2pt}
{\bf Attack Transferability --} One intriguing property of adversarial inputs is their transferability: an adversarial input effective against one \dnn is often found effective against another \dnn, though it is not crafted on the second one\mcite{papernot2016transferability,Liu:2017:iclr,moosavi:cvpr:2017}. In this set of experiments, we investigate whether such transferability exists in attacks against interpreters; that is, whether an adversarial input that generates a plausible interpretation against one interpreter is also able to generate a probable interpretation against another interpreter.

\begin{figure}[!ht]
 \hspace{-10pt}
 \epsfig{file = 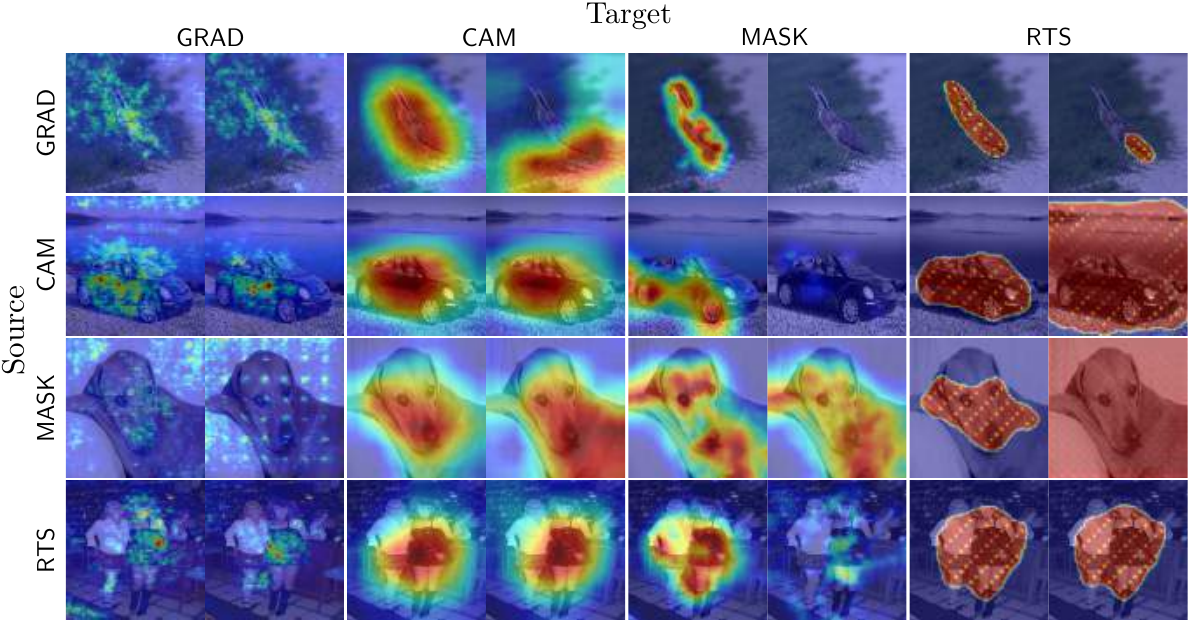, width=88mm}
 \caption{\small Visualization of attribution maps of adversarial inputs across different interpreters on \rnet. 	\label{fig:transfer}}
\end{figure}

Specifically, for each given interpreter $g$, we randomly select a set of adversarial inputs crafted against $g$ (source) and compute their interpretations on another interpreter $g'$ (target). Figure\mref{fig:transfer} illustrates the attribution maps of a given adversarial input on $g$ and $g'$. Further, for each case, we compare the adversarial map (right) against the corresponding benign map (left). Observe that the interpretation transferability is fairly low: an adversarial input crafted against one interpreter $g$ rarely generates highly plausible interpretation on another interpreter $g'$.

\begin{table}[!ht]{\footnotesize
 \centering
 \begin{tabular}{c|c|c|c|c}
	 & \grd & \cam & \mask & \rts\\
	 \hline
	 \hline
	 \grd &\cellcolor{Gray}$0.04$ & $0.24$ & $0.22$ & $0.24$\\
	 \cam    &$0.09$ & \cellcolor{Gray}$0.05$    &  $0.18$  &  $ 0.13$ \\
	 \mask   &$0.12$  & $0.34$  &    \cellcolor{Gray}$0.09$    & $ 0.74$      \\
	 \rts    &$0.10$ &  $0.17$   &  $0.20$                &  \cellcolor{Gray}$0.01$ \\

	 \hline
	 \pgd & $0.10$ & $0.22$     & $0.28$   & $0.22$ \\
 \end{tabular}
 \caption{\small $\mathcal{L}_1$ distance between attribution maps of adversarial (\system, \pgd) on \rnet (row/column as source/target). 	\label{tab:transferability}} }
\end{table}

We further quantitatively validate this observation. Table\mref{tab:transferability} measures the $\mathcal{L}_1$ distance between the adversarial and benign attribution maps across different interpreters. For comparison, it also shows the $\mathcal{L}_1$ measure for the adversarial inputs generated by \pgd. Observe that the adversarial inputs crafted on $g$ tends to generate low-quality interpretations on a different interpreter $g'$, with quality comparable to that generated by an interpretation-agnostic attack (i.e., \pgd). We can therefore conclude:

\begin{tcolorbox}[boxrule=0pt, title= Observation 7]
The transferability of adversarial inputs across different interpreters seems low.
\end{tcolorbox}

%
%
%

\vspace{2pt}
{\bf Attack Robustness --} It is observed in \msec{sec:eval} that the effectiveness of \system varies with the target interpreter. As shown in Figure\mref{fig:ioutest}, among all the interpreters, \system attains the lowest \iou scores on \grd, suggesting that \grd may be more robust against \system. This observation may be explained as follows: \grd uses the gradient magnitude of each input feature to measure its relevance to the model prediction; meanwhile, \system heavily uses the gradient information to optimize the prediction loss $\ell_\mathrm{prd}$; it is inherently difficult to minimize $\ell_\mathrm{prd}$ while keeping the gradient intact.

We validate the conjecture by analyzing the robustness of integrated gradient (\ig)\mcite{Sundararajan:2017:icml}, another back-propagation-guided interpreter, against \system. Due to their fundamental equivalence\mcite{Ancona:iclr:2018}, the discussion here also generalizes to other back-propagation interpreters (e.g.,\mcite{Simonyan:gradsaliency,Smilkov:iclr:2017,Shrikumar:2017:icml}).

%
%

At a high level, for the $i$-th feature of a given input $x$, \ig computes its attribution $m[i]$ by aggregating the gradient of $f(x)$ along the path from a baseline input $\bar{x}$ to $x$:
\begin{equation}
m[i] = (x[i] - \bar{x}[i]) \int_{0}^1 \frac{\partial f(tx + (1-t)\bar{x} )}{\partial x[i]}\mathrm{d}t
\end{equation}

Like other back-propagation interpretation models\mcite{Ancona:iclr:2018}, \ig satisfies the desirable completeness axiom\mcite{Shrikumar:2017:icml} that the attributions sum up to the difference between $f$'s predictions for the given input $x$ and the baseline $\bar{x}$.

%
%
%


To simplify the exposition, let us assume a binary classification setting with classes $\mathcal{C} = \{+, -\}$. The \dnn $f$ predicts the probability of $x$ belonging to the positive class as $f(x)$. Given an input $\bx$ from the negative class, the adversary attempts to craft an adversarial input $\ax$ to force $f$ to misclassify $\ax$ as positive. We define the prediction loss as $\ell_\mathrm{prd}(\ax) = f(\ax) - f(\bx)$ (i.e., the increase in the probability of positive prediction), which can be computed as:
\begin{equation}
	\label{eq:int}
\ell_\mathrm{prd}(\ax) = \int_{0}^1 \nabla f(t \ax + (1-t) \bx)^\matt (\ax - \bx)\mathrm{d}t
\end{equation}

Meanwhile, we define the interpretation loss as $\ell_\mathrm{int}(\ax) = \|\bbm - \am \|_1$, where $\bbm$ and $\am$ are the attribution maps of $\bx$ and $\ax$ respectively. While it is difficult to directly quantify $\ell_\mathrm{int}(\ax)$, we may use the attribution map of $\ax$ with $\bx$ as a surrogate baseline:
\begin{equation}
\Delta m[i]  =  (\ax[i] - \bx[i]) \int_{0}^1 \frac{\partial f(t\ax  + (1-t)\bx )}{\partial \ax[i]}\mathrm{d}t
\end{equation}
which quantifies the impact of the $i$-th input feature on the difference of $f(\bx)$ and $f(\ax)$. Thus, $\ell_\mathrm{int}(\ax) = \|\Delta m \|_1$.

\begin{prop}
With \ig, the prediction loss is upper bounded by the interpretation loss as:
 $\ell_\mathrm{prd}(\ax) \leq \ell_\mathrm{int}(\ax)$.
\end{prop}

\begin{proof}
We define $u$ as the input difference $u = (\ax - \bx)$ and $v$ as the integral vector with its $i$-th element $v[i]$ defined as
\begin{displaymath}
v[i] = \int_0^1\frac{\partial f(t\ax  + (1-t)\bx )}{\partial \ax[i]}\mathrm{d}t
\end{displaymath}
According to the definitions, we have
$\ell_\mathrm{prd}(\ax) = u^\matt v$ and
$\ell_\mathrm{int}(\ax) =  \| u \odot v \|_1$, where $\odot$ is the Hadamard product.

We have the following derivation:
$\ell_\mathrm{prd}(\ax)  = \sum_i u[i]v[i]  \leq \sum_i \| u[i]\cdot v[i] \| = \ell_\mathrm{int}(\ax)$.
Thus the prediction loss is upper-bounded by the interpretation loss.
%
\end{proof}
In other words, in order to force $\ax$ to be misclassified with high confidence, the difference of benign and adversarial attribution maps needs to be large. As the objectives of \system here is to maximize the prediction loss while minimizing the interpretation loss. The coupling between prediction and interpretation losses results in a fundamental conflict.

Note that however this conflict does not preclude effective adversarial attacks. First, the constraint of prediction and interpretation losses may be loose. Let $\gamma_\mathrm{prd}$ and $\gamma_\mathrm{int}$ be the thresholds of effective attacks. That is, for an effective attack, $\ell_\mathrm{prd}(\ax) \geq \gamma_\mathrm{prd}$ and $\ell_\mathrm{int}(\ax) \leq \gamma_\mathrm{int}$. There could be cases that $\gamma_\mathrm{prd} \ll \gamma_\mathrm{int}$, making \system still highly effective (e.g., Figure\mref{tab:iouSkinCancer}). Second, the adversary may pursue attacks that rely less on the gradient information to circumvent this conflict.


\vspace{2pt}
Overall, with the evidence of low attack transferability and disparate attack robustness, we can
conclude:
\begin{tcolorbox}[boxrule=0pt, title= Observation 8]
Existing interpreters tend to focus on distinct aspects of \dnn behavior, which may result in the prediction-interpretation gap.
\end{tcolorbox}

%% file: countermeasure.tex
\subsection*{Q3. Potential Countermeasures}
\label{sec:counter}

Based on our findings, next we discuss potential countermeasures against \system attacks.

\vspace{2pt}
{\bf Defense 1: Ensemble Interpretation --} Motivated by the observation that different interpreters focus on distinct aspects of \dnn behavior (e.g., \cam focuses on deep representations while \mask focuses on input-prediction correspondence), a promising direction to defend against \system is to deploy multiple, complementary interpreters to provide a holistic view of \dnn behavior.

Yet, two major challenges remain to be addressed. First, different interpreters may provide disparate interpretations (e.g., Figure\mref{fig:transfer}). It is challenging to optimally aggregate such interpretations to detect \system. Second, the adversary may adapt \system to the ensemble interpreter (e.g. optimizing the interpretation loss with respect to all the interpreters). It is crucial to account for such adaptiveness in designing the ensemble interpreter. We consider developing the ensemble defenses and exploring the adversary's adaptive strategies as our ongoing research directions.

%

%

\vspace{2pt}
{\bf Defense 2: Adversarial Interpretation --} Along the second direction, we explore the idea of adversarial training. Recall that \system exploit the prediction-interpretation gap to generate adversarial inputs. Here we employ \system as a drive to minimize this gap during training interpreters.

Specifically, we propose an {\em adversarial interpretation distillation} (\aid) framework. Let $\mathcal{A}$ be the \system attack. During training an interpreter $g$, for a given input $\bx$, besides the regular loss $\ell_\mathrm{map}(\bx)$, we consider an additional loss term $\ell_\mathrm{aid}(\bx) = - \|g(\bx) - g(\mathcal{A}(\bx)) \|_1$, which is the negative $\mathcal{L}_1$ measure between the attribution maps of $\bx$ and its adversarial counterpart $\mathcal{A}(\bx)$.
We encourage $g$ to minimize this loss during the training (details in Appendix A5).

To assess the effectiveness of \aid to reduce the prediction-interpretation gap, we use \rts as a concrete case study. Recall that \rts is a model-guided interpreter which directly predicts the interpretation of a given input. We construct two variants of \rts, a regular one and another with \aid training (denoted by \rtsa). We measure the sensitivity of the two interpreters to the underlying \dnn behavior.

\begin{figure}[!ht]
\centering
      	\epsfig{file = 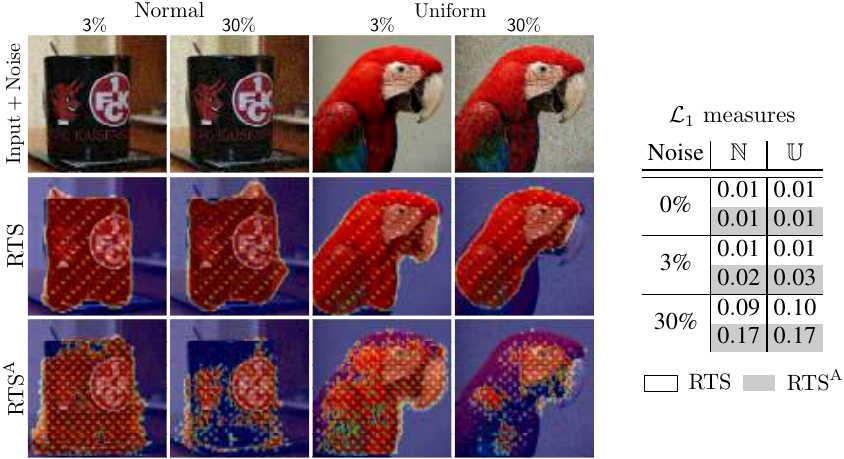, width=80mm}
    \caption{\small Attribution maps generated by \rts and \rtsa under different noise levels and types (normal $\mathbb{N}$, unifrom $\mathbb{U}$) on \rnet.
      \label{fig:aid-vis}}
\end{figure}

In the first case, we inject random noise (either normal or uniform) to the inputs and compare the attribution maps generated by the two interpreters. We consider two noise levels, which respectively cause 3\% and 30\% misclassification on the test set. Figure\mref{fig:aid-vis} shows a set of misclassified samples under the two noise levels.
Observe that compared  with \rts, \rtsa appears much more sensitive to the \dnn's behavior change, by generating highly contrastive maps. This sensitivity is also quantitatively confirmed by the $\mathcal{L}_1$ measures between clean and noisy maps on \rts and \rtsa.
The findings also corroborate a similar phenomenon observed in\mcite{Tsipras:2019:iclr}: the representations generated by robust models tend to align better with salient data characteristics.


\begin{figure}[!ht]
\centering
		\epsfig{file = 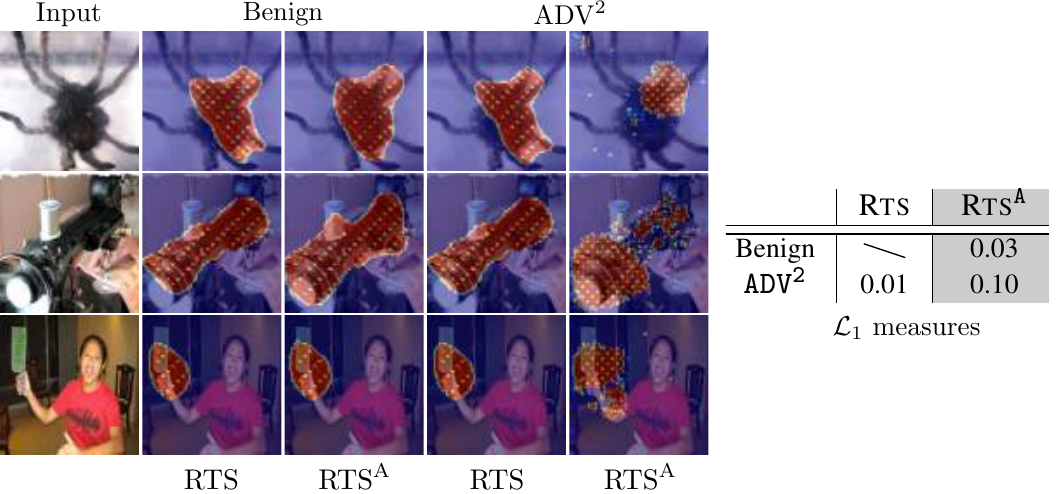, width=88mm}
		\caption{\small Attribution maps of benign and adversarial (\system) inputs with respect to \rts and \rtsa on \rnet.  \label{fig:aidv3}}
\end{figure}

In the second case, we assess the resilience of \rtsa against \system. In Figure\mref{fig:aidv3}, we compare the attribution maps of benign and adversarial inputs on \rts and \rtsa. It is observed that while \system generates adversarial inputs with interpretations fairly similar to benign cases on \rts, it fails to do so on \rtsa: the maps of adversarial inputs are fairly distinguishable from their benign counterparts. Moreover, \rtsa behaves almost identically to \rts on benign inputs, indicating that the \aid training has little impact on benign cases. These findings are confirmed by the $\mathcal{L}_1$ measures as well.

%
%
%

Overall we have the following conclusion.

\begin{tcolorbox}[boxrule=0pt, title= Observation 9]
It is possible to exploit \system to reduce the prediction-interpretation gap during training interpreters.
\end{tcolorbox}

%% file: literature.tex
\section{Related Work}
\label{sec:literature}


In this section, we survey three categories of work relevant to this work, namely, adversarial attacks and defenses, transferability, and interpretability.

\vspace{2pt}
{\bf Attacks and Defenses --}
Due to their widespread use in security-critical domains, machine learning models are increasingly becoming the targets of malicious attacks. Two primary threat models are considered in literature. Poisoning attacks -- the adversary pollutes the training data to eventually compromise the target models\mcite{Biggio:2012:icml}; Evasion attacks -- the adversary manipulates the input data during inference to trigger target models to misbehave\mcite{Dalvi:2004:kdd}.




Compared with simple models (e.g., support vector machines), securing deep neural networks (\dnns)  in adversarial settings entails more challenges due to their significantly higher model complexity\mcite{nat:dl}. One line of work focuses on developing new evasion attacks against {\dnns}\mcite{szegedy:iclr:2014,goodfellow:fsgm,madry:towards}. Another line of work attempts to improve \dnn resilience against such attacks by inventing new training and inference strategies\mcite{papernot:sp:2016,Xu:2018:ndss,Ma:2018:iclr}. Yet, such defenses are often circumvented by more powerful attacks\mcite{carlini:cw} or adaptively engineered adversarial inputs\mcite{Athalye:2018:iclr}, resulting in a constant arms race between attackers and defenders\mcite{Ling:2019:sp}.

This work is among the first to explore attacks against \dnns with interpretability as a means of defense.

\vspace{2pt}
{\bf Transferability --} One intriguing property of adversarial attacks is their transferability\mcite{szegedy:iclr:2014}: adversarial inputs crafted against one \dnn is often  effective against another one. This property enables black-box attacks -- the adversary generates adversarial inputs based on a surrogate \dnn and then applies them on the target model\mcite{papernot2016transferability,Liu:2017:iclr}. To defend against such attacks, the method of ensemble adversarial training\mcite{Tramer:2018:iclr} has been proposed, which trains \dnns using data augmented with adversarial inputs crafted on other models.

This work complements this line of work by investigating the transferability of adversarial inputs across different interpretation models.

\vspace{2pt}
{\bf Interpretability --} A plethora of interpretation models have been proposed to provide interpretability for black-box \dnns, using techniques based on back-propagation\mcite{Simonyan:gradsaliency,guidedbp:iclr:2015,Sundararajan:2017:icml}, intermediate representations\mcite{zhou:cam,selvaraju:gradcam,Du:arxiv:2018}, input perturbation\mcite{fong:mask}, and meta models\mcite{Dabkowski:nips:2017}.

The improved interpretability is believed to offer a sense of security by involving human in the decision-making process. Existing work has exploited interpretability to debug {\dnns}\mcite{Nguyen:2014:cvpr}, digest security analysis results\mcite{Guo:2018:ccs}, and detect adversarial inputs\mcite{Liu:kdd:2018,Tao:nips:2018}. Intuitively, as adversarial inputs cause unexpected \dnn behaviors, the interpretation of \dnn dynamics is expected to differ significantly between benign and adversarial inputs.

However, recent work empirically shows that some interpretation models seem insensitive to either \dnns or data generation processes\mcite{Adebayo:nips:2018}, while transformation with no effect on \dnns (e.g., constant shift) may significantly affect the behaviors of interpretation models\mcite{Kindermans:arxiv:2017}.

This work shows the possibility of deceiving \dnns and their coupled interpretation models simultaneously, implying that the improved interpretability only provides limited security assurance, which also complements prior work by examining the reliability of existing interpretation models from the perspective of adversarial vulnerability.

%% file: conclusion.tex
\section{Conclusion}
\label{sec:conclusion}

This work represents a systematic study on the security of interpretable deep learning systems (\imlses). We present \system, a general class of attacks that generate adversarial inputs not only misleading target \dnns but also deceiving their coupled interpretation models. Through extensive empirical evaluation, we show the effectiveness of \system against a range of \dnns and interpretation models, implying that the interpretability of existing \imlses may merely offer a false sense of security. We identify the prediction-interpretation gap as one possible cause of this vulnerability, raising the critical concern about the current assessment metrics of interpretation models. Further, we discuss potential countermeasures against \system, which sheds light on designing and operating \imlses in a more robust and informative fashion.

%% file: appendix.tex
\newpage
\section*{Appendix}
\label{sec:appendix}


\subsection*{A. Implementation Details}

\subsubsection*{A1. ADV$\bm{^2}$ against Grad-CAM}

Gradient-weighted class activation mapping (\gcam)\mcite{selvaraju:gradcam} is another feature-guided interpretation model. Similar to \cam, it approximates the model predictions as a weighted summation of the feature maps of the last convolutional layer. Differently, it projects global averaged gradients back to the convolutional feature maps:
\begin{equation}
	w_{k,c} = \frac{1}{Z} \sum_{i,j} \frac{ \partial f_c(x) }{ \partial a_k[i, j]  }
\end{equation}
where $f_c(x)$ is the model prediction with respect to input $x$ and class $c$, $a_k[i,j]$ is the activation of the $k$-th channel of the last convolutional layer at the spatial position $(i,j)$, and $Z$ is a normalization constant.
The attribution map is defined as:
	\begin{equation}
		m_c[i, j] = \mathrm{ReLU} \left( \sum_k w_{k,c} \, a_k[i, j] \right)
	\end{equation}

To attack \gcam, we apply the same optimization procedure in \msec{sec:cam}.
Note that although the gradient of $w_{k,y}(x)$ with respect to $x$ is zero almost everywhere, it is feasible to find high-quality solutions using stochastic gradient descent methods, since $a_k$ has non-zero gradients with respect to $x$.

\begin{table}[!ht]{\footnotesize
	\centering
	\begin{tabular}{c|c|c|c}
		             Attack  & ASR & $\Delta$ $(\mathcal{L}_1)$ & $\Delta$ $(\mathcal{L}_2)$ \\
    \hline
		\hline
	\texttt{P}    &  100.0\% (1.0)    &   $0.22$  & $ 0.28$   \\

\rowcolor{Gray}
		 \texttt{A}    &  100.0\% (0.98) & $0.07$  & $ 0.09$ \\
  \hline
	\end{tabular}
	\caption{\small Effectiveness of \pgd and \system against \gcam. \label{tab:grad-cam}}}
\end{table}

%

Table\mref{tab:grad-cam} summarizes the attack success rate, misclassification confidence, $\mathcal{L}_1$ and $\mathcal{L}_2$ distance between benign and adversarial maps (by \pgd and \system).

\subsubsection*{A2. Optimized MASK Formulation}

The complete optimization objective for \mask in \mcite{fong:mask} is given as follows:
\begin{align}
\nonumber \min_m & \quad \lambda_1 r_\mathrm{tv}(m) + \lambda_2 \| 1 -m \|_1
+ \mathbb{E}_\tau \left[ f_c \left( \phi ( x ( \cdot - \tau), m ) \right) \right] \\
\mathrm{s.t.} &  \quad 0 \leq m \leq 1
\label{opt:maskopt}
\end{align}
Here the term $r_\mathrm{tv}(m)$ is the total variation of $m$, which reduces its noise and artifacts; the term $\|1-m\|_1$ encourages the sparsity of $m$; $\phi(x;m)$ is the perturbation operator which blends $x$ with Gaussian noise (controlled by the parameter $\tau$); while $\lambda_1$ and $\lambda_2$ are the regularized coefficients for total variation and sparsity respectively.

\subsubsection*{\bf A3. Analysis of ADV$\bm{^2}$ against MASK}


Here we provide simple analysis for the effectiveness of \system in Algorithm\mref{alg:mask}.
We have the following proposition.

\begin{prop}
\label{lemma:simple}
 Let $f(x,y): \mathbb{R}^m \times \mathbb{R}^n \to \mathbb{R}$ be a function with continuous second-order derivative in the neighborhood $\mathcal{N} = \mathcal{N}_x
 \times \mathcal{N}_y$ of a point $(x_0, y_0)$. Assume the following conditions hold:
 \begin{mitemize}
 	\item $C_1$. For each $x \in \mathcal{N}_x$, there is a unique $g(x) \triangleq y \in \mathcal{N}_y$ such that $y$ is a local minimizer of $f(x, \cdot)$;
 	\item $C_2$. $y_0$ is a local minimizer of $f(x_0, \cdot)$ (i.e., $g(x_0) = y_0$);
 	\item $C_3$. The Hessian of $f(x_0, \cdot), H(x_0, \cdot) = \nabla_y^2 f(x_0, y_0)$ is non-degenerate at $y_0$ (i.e., $\det(H(x_0, y_0)) > 0$).
\end{mitemize}
Then for every $g_0 \in \mathcal{N}_y$, the gradient of $G(x) = \frac{1}{2} \| g(x) - g_0 \|_2^2$ at $x = x_0$ is given by \begin{equation}
	\label{eq:thmgrad}
	-\nabla_{xy} f( x_0, y_0 )(\nabla_y^2 f( x_0, y_0 ))^{-1} ( g(x_0) - g_0 ).
\end{equation}
\end{prop}

\begin{proof}
	 Let $Q$ be the partial derivative of $f$ with respect to $y$: $Q(x, y) \triangleq \nabla_y f(x, y)$. Then \meq{eq:thmgrad} is equivalent to
	 \begin{equation}
	 \label{eq:thmgradq}
	 	-\nabla_{x} Q( x_0, y_0 )(\nabla_y Q( x_0, y_0 ))^{-1} ( g(x_0) - g_0 )
	 \end{equation}

	 Since $g(x_0) = y_0$ (by $C_2$), $Q(x_0, y_0) = 0$. Based on $C_3$, for $J \triangleq \nabla_y Q(x_0, y_0)$, $\det(J) \neq 0 $.
	 According to the \textit{implicit function} theorem, there exists a neighborhood of $x_0$, $\hat{\mathcal{N}}_{x} \subset \mathcal{N}_{x} \subset \mathbb{R}^m $,
	 a neighborhood of $y_0$, $\hat{\mathcal{N}}_{y} \subset \mathcal{N}_{y} \subset \mathbb{R}^n $, and a unique smooth function $h(x): \hat{\mathcal{N}}_{x} \to \hat{\mathcal{N}}_{y}$ such that
	 \begin{equation}
	 \label{eq:defif}
	 Q(x, h(x)) = 0 \quad \forall x \in \hat{\mathcal{N}}_{x}
	 \end{equation}

According to $C_1$, for each $x \in \hat{\mathcal{N}}_{x}$, $f(x, \cdot)$ has a unique local minimizer $g(x)$ near $y_0$. Then the local minimizer must be $h(x)$ due to the first-order optimality condition. Thus $h(x) = g(x)$ for all $x \in \hat{\mathcal{N}_x}$. Computing the gradient with respect to $x$ for \meq{eq:defif}, we have
	 \begin{equation}
	 		\nabla_x g(x_0) = -\nabla_x Q( x_0, y_0 )(\nabla_y Q( x_0, y_0 ))^{-1}
	 \end{equation}
	 which leads to \meq{eq:thmgradq} by the product rule of gradient.
\end{proof}

Back to our case, let $\ell_\mathrm{map}(m;x)$ be the objective function defined in \meq{opt:maskopt4} (or \meq{opt:maskopt}) and $m_t$ be the target map. Although, strictly speaking, $\ell_\mathrm{map}(m;x)$ is not continuously differentiable, we assume $f(x, m) = \ell_\mathrm{map}(m;x)$ satisfies the conditions of Proposition\mref{lemma:simple} for analysis purpose.

Note that $m_*(x) = \arg\min_m \ell_\mathrm{map}(m;x)$ is the optimal map found by \mask for given $x$.
 At a given iteration, let $m_*$ be the optimal map of the current input $x$. We can take a step towards the direction of
 \begin{equation}
\nonumber
\varDelta =  -\nabla_{xm} \ell_\mathrm{map}(m_*;x)(\nabla_m^2 \ell_\mathrm{map}(m_*;x))^{-1} (m_* - m_t )
\end{equation} to make the map generated in the next iteration approach $m_t$. In practice we only have $\tilde{m}$, an estimate of $m_*$. Let $H = \nabla_m^2 \ell_\mathrm{map}(\tilde{m};x)$ be the Hessian matrix for fixed $x$. By plugging $\tilde{m}$ into $\varDelta$ and setting a proper step size $\alpha > 0$, we have
\begin{align}
\alpha \varDelta \approx & -\alpha\nabla_{xm} \ell_\mathrm{map}(\tilde{m};x)(\nabla_m^2 \ell_\mathrm{map}(\tilde{m};x))^{-1} (\tilde{m} - m_t )
\nonumber \\
 = &  -\alpha \nabla_{xm} \ell_\mathrm{map}(\tilde{m};x) H^{-1}(\tilde{m} - m_t) \nonumber \\
  = & \nabla_{x} \underbrace{(\tilde{m}  -\alpha H^{-1} \nabla_m \ell_\mathrm{map}(\tilde{m};x))}_\text{inner update step}  (\tilde{m} - m_t) \nonumber
  \label{eq:deltaapprox}
 \end{align}
 In our attack, we instantiate the inner update step with an Adam step to get fast convergence and to avoid the issue of vanishing Hessian for \dnns with ReLU activation.

\subsubsection*{\bf A4: Details of StAdv-based ADV$\bm{^2}$}

We first briefly introduce the concept of spatial transformation. Let $\tilde{x}_i$ be the $i$-th pixel of adversarial input $\tilde{x}$ and $(\tilde{u}_i, \tilde{v}_i)$ be its spatial coordinates. With flow-based transformation, $\tilde{x}$ is generated from another input $x$ by a per-pixel flow vector $r$, where $r_i = (\Delta u_i, \Delta v_i)$. The corresponding coordinates of $\tilde{x}_i$ in $x$ are given by $(u_i, v_i) = (\tilde{u}_i + \Delta u_i,\tilde{v}_i + \Delta v_i)$. As $(u_i, v_i)$ do not necessarily lie on the integer grid, bilinear interpolation\mcite{jaderberg:2015:nips} is used to compute $\tilde{x}_i$:
\begin{equation}
\nonumber
 \hspace*{-5pt}  \tilde{x}_i  = \sum_{j}  x_j   \max(0, 1 - |\tilde{u}_i  + \Delta u_i  - u_j |)    \max(0, 1 - |\tilde{v}_i  + \Delta v_i - v_{j} |)
\end{equation}
where $j$ iterates over the pixels adjacent to $(u_i, v_i)$ in $x$. With \adef as the underlying attack framework, \system can be constructed as optimizing the following objective:
\begin{equation}
 \hspace*{-10pt}  \min_r \, \ell_\mathrm{prd}(f(x+r), c_t) +  \lambda \ell_\mathrm{int}(g(x + r;f), m_t) +  \tau \ell_\mathrm{flow}(r)
	\label{eq:stadvloss}
\end{equation}
where $\ell_\mathrm{flow}(r) = \sum_i \sum_{j \in \mathcal{N}(i)} \sqrt{\| \Delta u_i - \Delta u_j \|_2^2 + \| \Delta v_i - \Delta v_j\|_2^2} $ measures the magnitude of spatial transformation and $\tau$ is a hyper-parameter controlling its importance. In implementation, we solve \meq{eq:stadvloss} using an Adam optimizer.

\subsubsection*{\bf A5: Details of AID}

We use \rts as a concrete example to show the implementation of \aid. In \rts, one trains a \dnn $g$ (parameterized by $\theta$) to directly predict the attribution map $g(x;\theta)$ for a given input $x$. To train $g$, one minimizes the interpretation loss:
\begin{align}
\ell_\mathrm{int}(\theta) \triangleq  & \lambda_1 r_\mathrm{tv}(g(x;\theta)) + \lambda_2 r_\mathrm{av}(g(x;\theta))  - \log \left( f_c \left( \phi (x; g(x;\theta)) \right) \right) \nonumber \\
&	  + \lambda_3 f_c \left( \phi ( x; 1 - g(x;\theta) ) \right)^{\lambda_4}
\end{align}
with all the terms defined similarly as in \meq{opt:rtsinit}.

In \aid, let $\mathcal{A}$ denote the \system attack. We further consider an adversarial distillation loss:
\begin{align}
  \ell_\mathrm{aid}(\theta) \triangleq -\| g(x;\theta) -  g(\mathcal{A}(x);\theta)  \|_1
\end{align}
which measures the difference of attribution maps of benign and adversarial inputs under the current interpreter $g(\cdot;\theta)$.

\aid trains $g$ by alternating between minimizing $\ell_\mathrm{int}(\theta)$ and minimizing $\ell_\mathrm{aid}(\theta)$ until convergence.

\subsection*{B. Parameter Setting}

Here we summarize the default parameter setting for the attacks implemented in this paper.


\subsubsection*{B1. PGD-based ADV$\bm{^2}$}

For regular \pgd, we set the learning rate $\alpha = 1./255$ and the perturbation threshold $\epsilon = 0.031$.
Table\mref{tab:hplpnorm} list the parameter setting of \pgd-based \system.

\begin{table}[!ht]{\footnotesize
\centering
		\begin{tabular}{c|l|l|l}
			& Parameter  & \rnet & \dnet  \\
			\hline \hline
			\multirow{2}[0]{*}{\grd} & \# iterations $n_\mathrm{total}$ & 800   & 800  \\
			&  	$\ell_\mathrm{int}$ coefficient ($\lambda$) & 0.007  & 0.007 \\
			\hline
			\multirow{2}{*}{\cam} & \# iterations $n_\mathrm{total}$ & 1200   & 1200  \\
			&  	$\ell_\mathrm{int}$ coefficient ($\lambda$) & 0.204/0.02  & 0.204 \\
			\hline
			\multirow{5}{*}{\mask} & \# iterations $n_\mathrm{total}$  & 1000  & 1000 \\
			& \# gradient descent steps $n_\mathrm{step}$ & 4     & 4  \\
			& \# iterations per reset $n_\mathrm{reset} $ & 50    & 50  \\
			& max. search step size $\alpha_\mathrm{max}$  & 0.08 &  0.08 \\
			& \# max. search steps  $n_\mathrm{bs}$ & 12 & 12   \\
			\hline
			\multirow{3}{*}{\rts} & \# iterations $n_\mathrm{total}$ & 1200   & 1200  \\
			& $\ell_\mathrm{int}$ coefficient ($\lambda_1$)  & 0.002/0.006 & 0.008  \\
			& $\ell_\mathrm{prd}$ coefficient ($\lambda_2$) & 0.1/-    & 0.1  \\
			\hline
		\end{tabular}%
		\caption{\small Parameter setting of \pgd-based \system. The setting for Q4 in \msec{sec:eval} (if different) is shown after `/'.
			\label{tab:hplpnorm}}}
	\end{table}%

\subsubsection*{\bf B2. StAdv-based ADV$\bm{^2}$}

Table\mref{tab:hpstadv} list the parameter setting of \adef-based \system.

	\begin{table}[!ht]{\footnotesize
		\centering
		\begin{tabular}{c|l|l|l}
	& Parameter  & \rnet & \dnet \\
			\hline \hline
			\multirow{3}[0]{*}{\grd} & \# iterations $n_\mathrm{total} $& 800   & 800 \\
			& $\ell_\mathrm{int}$ coefficient ($\lambda$) & 0.023 & 0.023 \\
			& $\ell_\mathrm{flow}$ coefficient ($\tau$) & 0.0005 & 0.0005 \\
			\hline
			\multirow{3}[0]{*}{\cam}
			& \# iterations $n_\mathrm{total} $& 600   & 600 \\
			& $\ell_\mathrm{int}$ coefficient ($\lambda$) & 0.653 & 0.653 \\
			& $\ell_\mathrm{flow}$ coefficient ($\tau$) & 0.0005 & 0.0005 \\
			\hline
			\multirow{5}[0]{*}{\mask} & \# iterations $n_\mathrm{total}$ & 1000  & 1000 \\
			& \# gradient descent steps $n_\mathrm{step}$ & 4     & 4 \\
			& \# iterations per reset $n_\mathrm{reset}$ & 50    & 50 \\
			& $\ell_\mathrm{int}$ coefficient ($\lambda$) & 500   & 500 \\
			& $\ell_\mathrm{flow}$ coefficient ($\tau$) & 0.004 & 0.005 \\
			\hline
			\multirow{4}[0]{*}{\rts} & \# iterations $n_\mathrm{total}$ & 600   & 600 \\
			& $\ell_\mathrm{int}$ coefficient ($\lambda_1$)  & 0.0408 & 0.0612 \\
			& $\ell_\mathrm{prd}$ coefficient ($\lambda_2$) & 0.1   & 0.1 \\
			& $\ell_\mathrm{flow}$ coefficient ($\tau$) & 0.0005 & 0.0005 \\
			\hline
		\end{tabular}%
		\caption{\small Parameter setting of \adef-based \system. \label{tab:hpstadv} }}
	\end{table}%

The Adam optimizer in our experiments uses the hyper-parameter setting of $(\alpha, \beta_1, \beta_2) = (0.01, 0.9, 0.999)$.


\subsection*{C. Additional Experimental Results}

Below we include more experimental results that complement the ones presented in \msec{sec:eval} and \msec{sec:tradeoff}.

\subsubsection*{\bf C1. \msec{sec:eval} Q2 Attack Effectiveness (Interpretation)}

Figure\mref{fig:ensemble-sample-den} shows a set of sample inputs (benign and adversarial) and their attribution maps generated by \grd, \cam, \mask, and \rts on \dnet.

\begin{figure}[!ht]
	\hspace{-10pt}
	\epsfig{file = 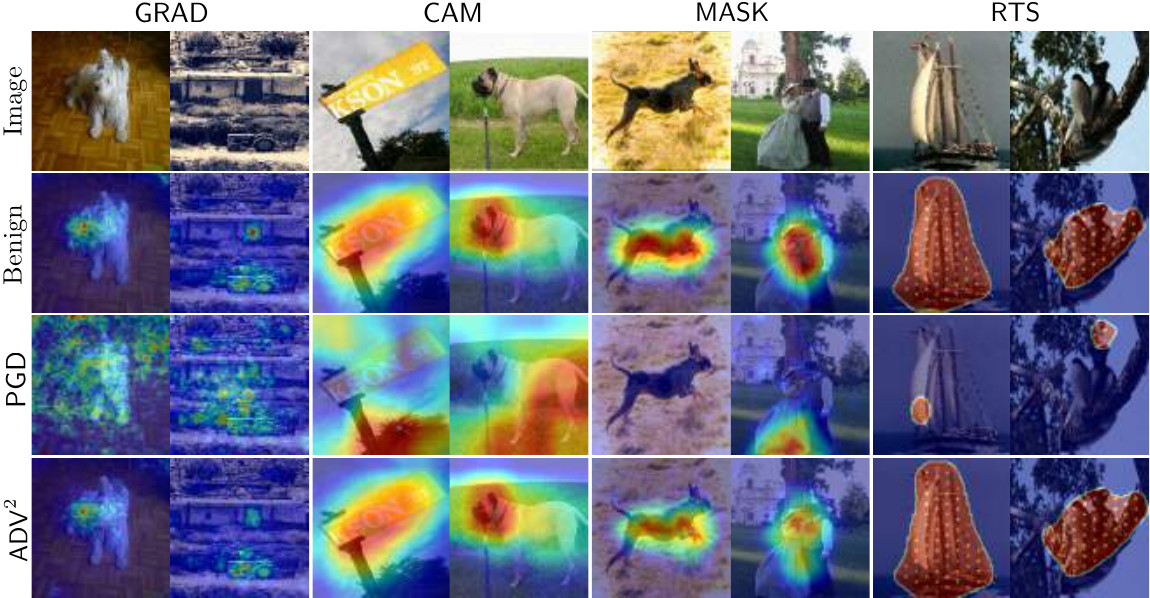, width=90mm}
	\caption{\small Attribution maps of benign and adversarial (\pgd, \system) inputs with respect to \grd, \cam, \mask, and \rts on \dnet. 	\label{fig:ensemble-sample-den}}
\end{figure}

Table\mref{tab:lpdistanceAppendix} lists the average $\mathcal{L}_\mathrm{p}$ distance ($p = 1,2$) between the attribution maps of benign and adversarial (\pgd, \system) inputs, which complements the results in Figure\mref{tab:lpdistance}. We normalize the $\mathcal{L}_2$ measures by dividing them by the square root of the number of pixels.

\begin{table}[!ht]{\footnotesize
	\centering
	\begin{tabular}{c|c|c|c|c|c}
	&	\multirow{2}{*}{Attack} & \multicolumn{2}{c|}{\rnet} &\multicolumn{2}{c}{\dnet}\\
		\cline{3-6}
		& & $\Delta$ $(\mathcal{L}_1)$ & $\Delta$ $(\mathcal{L}_2)$ & $\Delta$ $(\mathcal{L}_1)$ & $\Delta$ $(\mathcal{L}_2)$\\
		\hline
		\hline
		 		\multirow{2}{*}{\grd} & \texttt{P}  & $0.10$ & $0.14$ & $0.11$ & $0.15$ \\
		& \cellcolor{Gray} \texttt{A}		& \cellcolor{Gray}  $0.04$ & \cellcolor{Gray}  $0.07$ & \cellcolor{Gray} $0.05$ &\cellcolor{Gray}  $0.07$ \\
		\hline
	\multirow{2}{*}{\cam} &  \texttt{P}       & $0.22$  & $0.28$   & $0.31$   & $0.36 $  \\
	&	\cellcolor{Gray} \texttt{A}  & \cellcolor{Gray} $0.05$  & \cellcolor{Gray} $0.06$   &\cellcolor{Gray}  $0.04$   & \cellcolor{Gray} $0.05$  \\
		\hline
	\multirow{2}{*}{\mask}  &  \texttt{P}     & $0.28$  & $0.38$   & $0.27$   & $0.37$  \\
	&  \cellcolor{Gray} \texttt{A}     & \cellcolor{Gray} $0.09$  & \cellcolor{Gray} $0.16$   & \cellcolor{Gray} $0.09$   & \cellcolor{Gray} $0.17$  \\
		\hline
		\multirow{2}{*}{\rts} &  \texttt{P}  & $0.22$  & $0.42$   & $0.26$   & $0.48$  \\
	&	\cellcolor{Gray} \texttt{A}     &\cellcolor{Gray}  $0.01$  & \cellcolor{Gray} $0.06$   &\cellcolor{Gray}  $0.02$   & \cellcolor{Gray} $0.09$  \\
		\hline
	\end{tabular}
	\caption{\small $\mathcal{L}_\mathrm{p}$ distance between attribution maps of benign and adversarial (P-\pgd, A-\system) inputs. \label{tab:lpdistanceAppendix}}}
\end{table}


\subsubsection*{\bf C2. \msec{sec:eval} Q4 Real Application}

In \msec{sec:eval}, we use the dataset from the ISIC 2018 challenge\footnote{\url{https://challenge2018.isic-archive.com/task3/}},
and adopt a competition-winning model\mcite{gessert:2018:arxiv} (with \rnet as its backbone), which achieves the second place in the challenge. The confusion matrix in Figure\mref{fig:confusion_matrix} shows the performance of the classifier in our study.

\begin{figure}[!ht]
	\centering
	\epsfig{file = 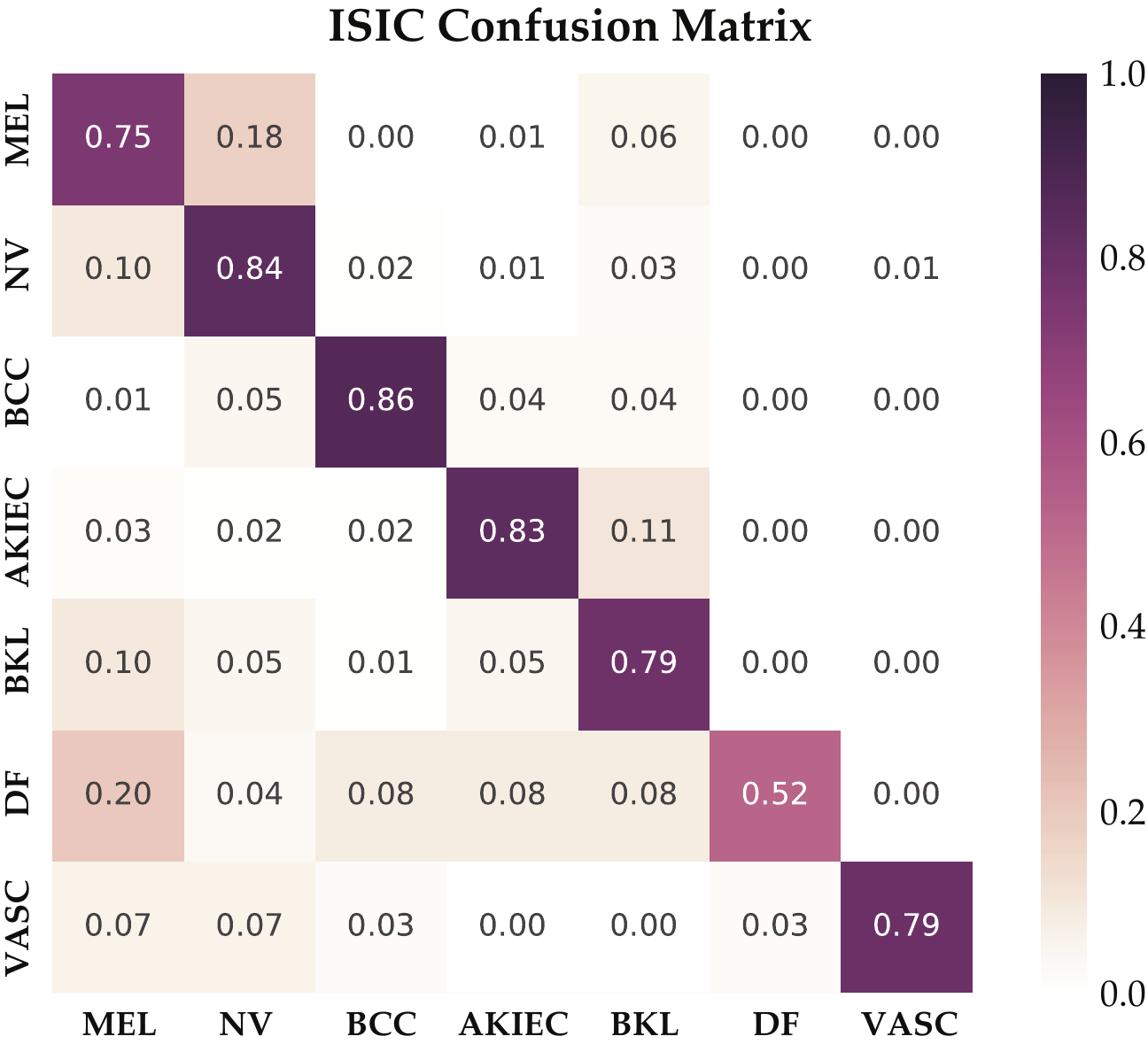, width=55mm}
	\caption{\small Confusion matrix of the classifier used in \msec{sec:eval} Q4 on the ISIC 2018 challenge dataset\mcite{gessert:2018:arxiv}.	\label{fig:confusion_matrix}}
\end{figure}


Table\mref{tab:lpdistanceskin} lists the average $\mathcal{L}_\mathrm{p}$ distance ($p = 1,2$) between the attribution maps of benign and adversarial (\pgd, \system) inputs in the case study of skin cancer diagnosis.

\begin{table}[!ht]{\footnotesize
	\centering
	\begin{tabular}{c|c|c|c}
	&  Attack  	& $\Delta$ $(\mathcal{L}_1)$ & $\Delta$ $(\mathcal{L}_2)$ \\
		\hline
		\hline
		\multirow{2}{*}{\grd} & \texttt{P} & $0.19$ & $0.23$ \\
		& \cellcolor{Gray} \texttt{A} 	& \cellcolor{Gray} $0.06$ & \cellcolor{Gray} $0.09$ \\
		\hline
		\multirow{2}{*}{\cam} &  \texttt{P}   & $0.25$  & $0.31$ \\
	&	\cellcolor{Gray} \texttt{A}   & \cellcolor{Gray} $0.06$  & \cellcolor{Gray} $0.08$\\
		\hline
	\multirow{2}{*}{\mask}  &  \texttt{P} & $0.23$  & $0.30$ \\
		&  \cellcolor{Gray} \texttt{A}     & \cellcolor{Gray} $0.08$  & \cellcolor{Gray} $0.11$ \\
		\hline
		\multirow{2}{*}{\rts} &  \texttt{P} & $0.15$  & $0.26$\\
	&	\cellcolor{Gray} \texttt{A}    & \cellcolor{Gray} $0.02$  & \cellcolor{Gray}$0.07$ \\
		\hline
	\end{tabular}
	\caption{\small $\mathcal{L}_\mathrm{p}$ distance of attribution maps of benign and adversarial (\pgd, \system) inputs in the case study of skin cancer diagnosis.
		\label{tab:lpdistanceskin}}}
\end{table}

\subsubsection*{\bf C3. \msec{sec:eval} Q5 Alternative Attack Framework}

Figure\mref{fig:ensemble-sample-Dense-st} visualizes attribution maps of benign and adversarial (\adef, \adef-based \system) inputs on \dnet. Figure\mref{fig:stadv-iou} further compares the $\mathcal{L}_1$ measures and \iou scores (w.r.t. benign cases) of adversarial inputs on \dnet.

\begin{figure}[!ht]
\epsfig{file = 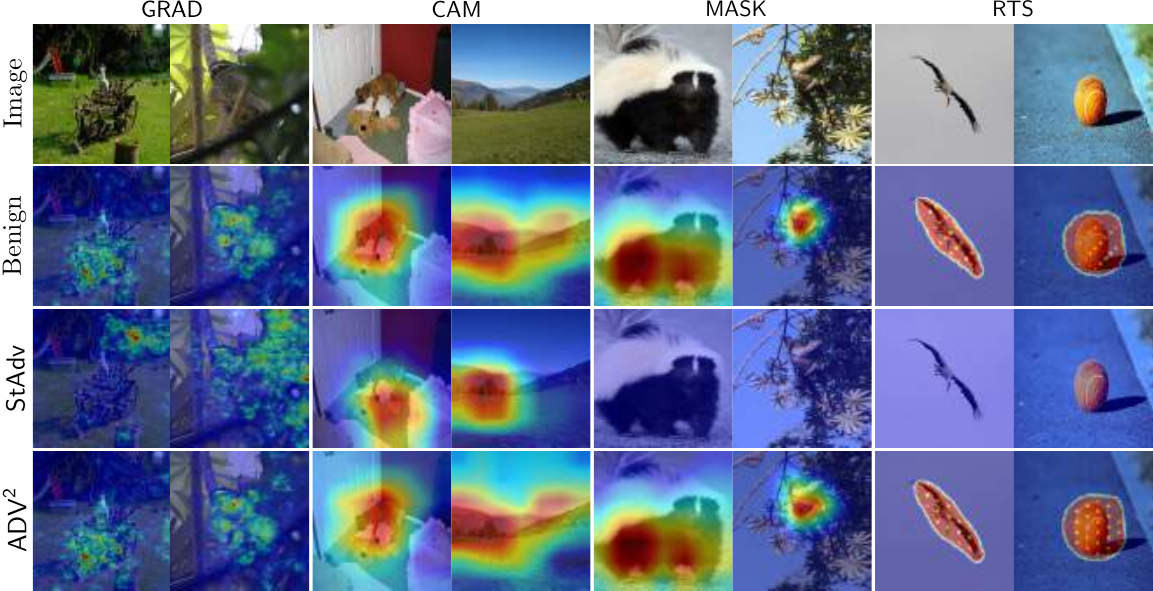, width=85mm}
\caption{\small Attribution maps of benign and adversarial (\adef and \adef-based \system) inputs on \dnet. \label{fig:ensemble-sample-Dense-st}}
\end{figure}

\begin{figure}[!ht]
	\epsfig{file = 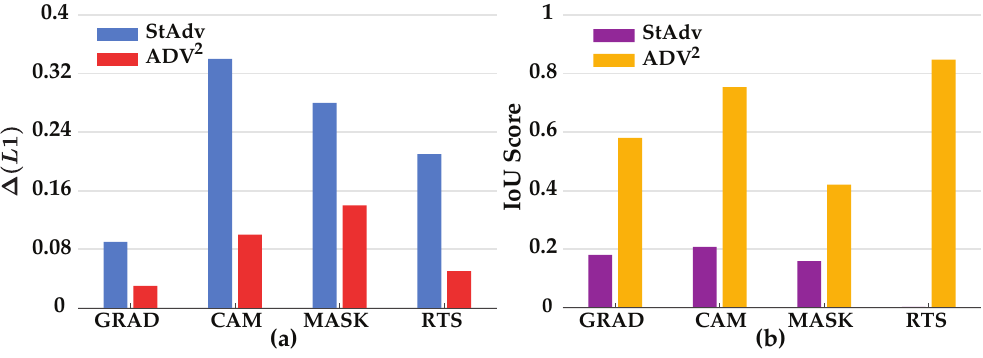, width=85mm}
	\caption{\small$\mathcal{L}_1$ measures and \iou scores of adversarial (\adef, \adef-based \system) inputs w.r.t. benign maps on \dnet.
		\label{fig:stadv-iou}}
\end{figure}

\subsubsection*{\bf C4. \msec{sec:tradeoff} Q1 Random Class Interpretation}

Figure\mref{fig:target-map-den} visualizes attribution maps of target and adversarial (\system) inputs on \dnet, which complements the results shown in Figure\mref{fig:target-map-res}. Figure\mref{fig:ioutest-target} compares the $\mathcal{L}_1$ measures and \iou scores of adversarial maps w.r.t. benign and target cases on \dnet.

\begin{figure}[!ht]
	\hspace{-10pt}
	\epsfig{file = 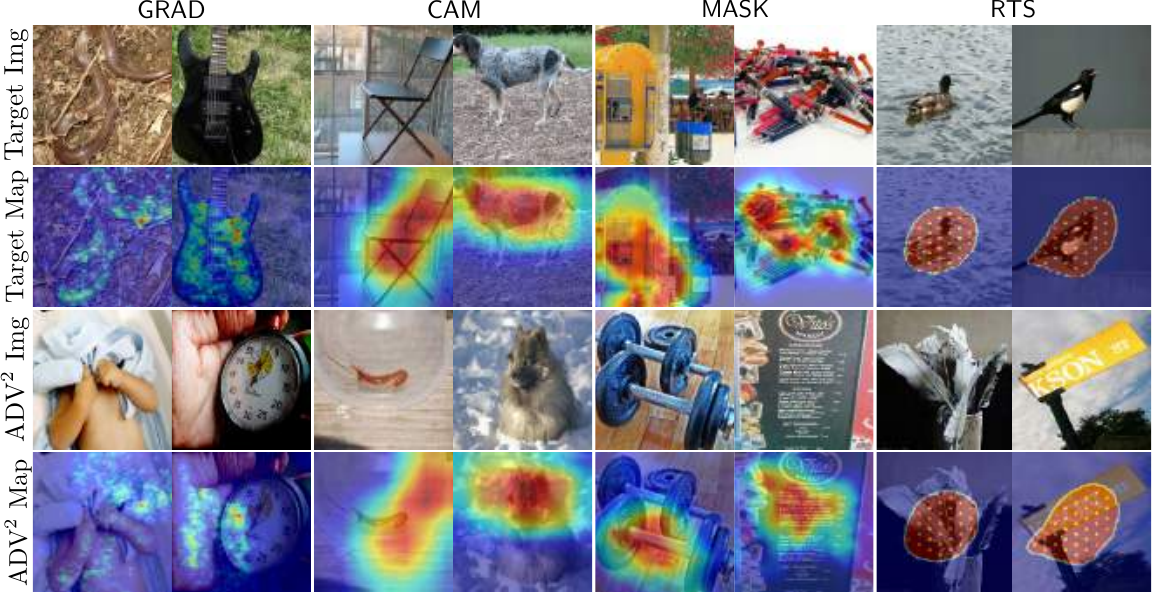, width=85mm}
	\caption{\small Target and adversarial (\system) inputs and their attribution maps on \dnet.
\label{fig:target-map-den}}
\end{figure}

\begin{figure}[!ht]
	\epsfig{file = 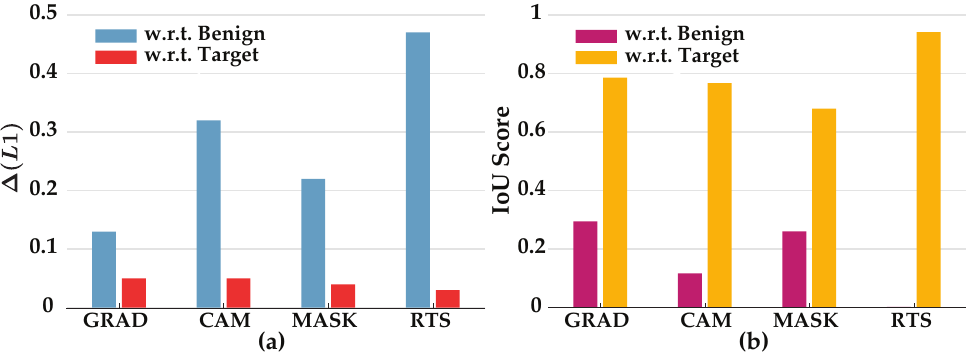, width=85mm}
	\caption{\small $\mathcal{L}_1$ measures (a) and IoU scores (b) of adversarial maps with respect to benign and target cases on \dnet.
		\label{fig:ioutest-target}}
\end{figure}